\documentclass{article}
\usepackage{graphicx}
\usepackage{amsfonts, amsmath, amsthm, amssymb}
\usepackage{booktabs}
\usepackage{xcolor}
\usepackage{mathtools}
\usepackage{setspace} 
\usepackage[left=1in,top=1in,right=1in,bottom=1in]{geometry}
\usepackage{multirow}
\usepackage{placeins}
\usepackage{algorithm}
\usepackage{algorithmic}
\usepackage[round]{natbib}
\usepackage[colorlinks=true, linkcolor=blue, citecolor=blue, urlcolor=blue]{hyperref}
\usepackage{pifont} 

\newtheorem{remark}{Remark}
\newtheorem{theorem}{Theorem}[section]

\newtheorem{lemma}[theorem]{Lemma}

\newcommand{\RR}[0]{\mathbb{R}}

\newcommand{\Bern}[0]{\text{Bern}}

\newcommand{\Var}[0]{\text{Var}}

\newcommand{\convp}[0]{\xrightarrow[]{P}}
\newcommand{\norm}[1]{\left\lVert#1\right\rVert}

\title{Efficient Targeted Maximum Likelihood Estimators\\ for Two-Phase Design Problems}
\author{Sky Qiu$^{0,1,2}$, Susan Gruber$^2$, Pamela A. Shaw$^{3,4}$, Brian D. Williamson$^{3,4}$, \\and 
Mark J. van der Laan$^1$}
\date{\today}

\begin{document}
\maketitle
\begin{abstract}
In a typical two-phase design, a random sample is drawn from the target population in phase 1, during which only a subset of variables is collected. In phase 2, a subsample of the phase-1 cohort is selected, and additional variables are measured. This setting induces a coarsened data structure on the data from the second phase. We assume coarsening at random, that is, the phase-2 sampling mechanism depends only on variables fully observed. Within the established semiparametric theory for two-phase and coarsened data, we consider estimation of full-data parameters of interest. We review existing estimators, including the generalized raking estimator and the inverse probability of censoring weighted targeted maximum likelihood estimation (IPCW-TMLE) along with its extensions that also target the phase-2 sampling mechanism to improve efficiency. We further introduce a new class of estimators constructed within the TMLE framework that are asymptotically equivalent.
\end{abstract}

\section{Introduction}\label{sec:introduction}
\footnotetext[0]{Correspondence: \tt sky.qiu@berkeley.edu}
\footnotetext[1]{Division of Biostatistics, School of Public Health, University of California, Berkeley, CA, USA}
\footnotetext[2]{Targeted ML Solutions Inc., Cambridge, MA, USA}
\footnotetext[3]{Biostatistics Division, Kaiser Permanente Washington Health Research Institute, Seattle, WA, USA}
\footnotetext[4]{Department of Biostatistics, University of Washington, Seattle, WA, USA}

We consider a general two-phase design in which the full data are denoted by $X\sim P_{X,0}\in\mathcal{M}^F$, where $P_{X,0}$ is the probability distribution for the full data in a nonparametric statistical model $\mathcal{M}^F$. In phase 1, only a subset of the full data is observed for all individuals, while in phase 2, a subsample is selected for additional data collection, so that the full data are only available on the individuals in phase 2. The observed data can be represented as $O=(Z,\Delta,\Delta X)$, where $Z\subset X$ denotes the components of the full data observed on everyone, and $\Delta\in\{0,1\}$ indicates phase-2 membership. We assume coarsening at random (CAR) \citep{heitjan_ignorability_1991,gill_coarsening_1997}, meaning that the phase-2 sampling indicator $\Delta$ depends only on $Z$, which is fully observed. We observe $n$ independent and identically distributed copies of $O$ drawn from $P_0\in\mathcal{M}$. The observed-data model $\mathcal{M}$ is induced by $\mathcal{M}^F$ and the allowed sampling mechanisms under CAR and sampling positivity, as defined in Section~\ref{sec:notation}.

As a motivating example, we specialize this setup to a point-treatment setting. Let $X=(W,A,Y)$, where $W=(W_1,W_2)\in\mathcal{W}\subset\mathbb{R}^d$ are baseline covariates, $A\in\{0,1\}$ is a binary treatment, and $Y\in\mathcal{Y}\subset\mathbb{R}$ is an outcome. In phase 1, we observe $V=(W_1,A,Y)$ for everyone. In phase 2, we collect additional baseline covariates $W_2$ for a subsample, giving observed data $O=(V,\Delta,\Delta W_2)$. Here, ``full data'' refers to the uncoarsened variables $X$ that would be available if $W_2$ were measured for everyone. It does not include both potential outcomes.

For illustration, we consider estimating the full-data statistical target parameter $\Psi^F:\mathcal{M}^F\rightarrow\RR$ defined as the covariate-standardized mean difference
\begin{equation}\label{eq:full_data_target}
\Psi^F(P_X)\coloneq\int\{Q_{P_X}(1,w)-Q_{P_X}(0,w)\}\,dP_{X,W}(w),
\end{equation}
where $Q_{P_X}(a,w)\coloneq E_{P_X}(Y\mid A=a,W=w)$ and $P_{X,W}$ is the marginal distribution of $W$ under $P_X$. We take $\mathcal{M}^F$ to be nonparametric subject to treatment positivity, $0<g_{P_X}(W)<1$ almost surely under $P_{X,W}$, and integrability of $Q_{P_X}(a,W)$ for $a\in\{0,1\}$. Here $g_{P_X}(w)\coloneq P_X(A=1\mid W=w)$, so $1-g_{P_X}(w)=P_X(A=0\mid W=w)$. These conditions make the statistical functional well defined. Its true value is $\psi_0\coloneq\Psi^F(P_{X,0})$.

For a causal interpretation, let $Y_a$ denote the potential outcome under treatment $a$ \citep{rubin_estimating_1974}. Assume well-defined treatments with no interference, consistency $Y=Y_a$ when $A=a$, and conditional exchangeability $(Y_0,Y_1)\perp A\mid W$. Together with treatment positivity, these assumptions identify the average treatment effect (ATE), denoted by $\theta_0\coloneq E(Y_1-Y_0)$, through
\begin{equation}\label{eq:causal_identification}
\theta_0=\Psi^F(P_{X,0}).
\end{equation}
The expectation defining $\theta_0$ is under the causal distribution of the potential outcomes, whereas $P_{X,0}$ is the distribution of $(W,A,Y)$. Without the causal assumptions, $\Psi^F(P_{X,0})$ remains a statistical target but need not equal the ATE \citep{hines_demystifying_2022}. Our methods estimate this statistical target from two-phase data. They also apply to standardized treatment-specific means $\int Q_{P_X}(a,w)\,dP_{X,W}(w)$ and their contrasts, with causal interpretations under the same assumptions, and can be extended to other identifiable full-data statistical parameters.

For example, consider an observational study of a type 2 diabetes medication and the one-year risk of a major adverse cardiac event (MACE). In phase 1, we identify patients eligible for the medication and collect baseline covariates $W_1$, such as age, sex, and race or ethnicity, from electronic health records. Treatment is not randomized. It is chosen through clinical care, with patients receiving either the medication ($A=1$) or a comparator ($A=0$). In phase 2, we sample patients for additional pretreatment measurements $W_2$, such as biomarkers, diet, and smoking history. We may oversample patients with high cardiovascular risk based on $W_1$, while retaining a positive sampling probability for all covariate strata. Here $W_2$ represents potential confounders that may affect both treatment choice and MACE risk. It is not assumed to be a surrogate for $W_1$. The causal interpretation requires exchangeability given $(W_1,W_2)$. These measurements are expensive and are collected only for the phase-2 subsample. All patients are assessed for MACE at one year, giving $O=(W_1,\Delta,\Delta W_2,A,Y)$.

Other examples include nested case-control studies, where phase-2 sampling depends on $Y$ \citep{ernster_nested_1994}. A different use of the covariate partition appears in the surrogate marker problem in Example~1 of Chapter~14 of \citet{tsiatis_semiparametric_2006}. There, $W_1$ is a cheaper surrogate for an expensive biomarker $W_2$.

\begin{table}[ht]
\centering
\begin{tabular}{|ccccc|}
\hline
$W_1$ & $W_2$ & $A$ & $Y$ & $\Delta$ \\
\hline
\checkmark & \ding{55} & \checkmark & \checkmark & \multirow{3}{*}{$0$} \\
\vdots     & \vdots    & \vdots     & \vdots     &                      \\
\checkmark & \ding{55} & \checkmark & \checkmark &                      \\
\hline
\checkmark & \checkmark & \checkmark & \checkmark & \multirow{3}{*}{$1$} \\
\vdots     & \vdots     & \vdots     & \vdots     &                      \\
\checkmark & \checkmark & \checkmark & \checkmark &                      \\
\hline
\end{tabular}
\caption{Observed data structure for an example of a two-phase design. Here, ``\checkmark'' denotes an observed value, and ``\ding{55}'' denotes a missing value. In phase-1, we observe $V=(W_1,A,Y)$ for all individuals, while in phase-2 we additionally observe $W_2$ for a subsample with $\Delta=1$.}
\label{table:twophase}
\end{table}

Efficient estimation under two-phase sampling has a broader history than the particular estimators compared here. For example, \citet{gilbert_optimal_2014} study efficient estimation of means and mean differences and optimize phase-2 sampling using auxiliary covariates, costs, and conditional outcome variability. We treat the sampling design as given and focus on estimation when some baseline covariates are measured only in phase 2. Related double-sampling methods recover outcomes for a subsample of individuals with initially missing outcomes. In electronic health record studies, \citet{levis_double_2024} allow that initial missingness to be informative, with identification based on conditional independence and positivity assumptions for subsequent outcome ascertainment. That setting differs from our CAR assumption for missing baseline covariates.

These settings share tools from semiparametric theory for coarsened data. For missing exposure information, \citet{kennedy_missing_2020} develops efficient inference and relates the observed-data influence curve to the full-data curve through the mapping of \citet{robins_estimation_1994} and Theorem~7.2 of \citet{tsiatis_semiparametric_2006}. For stochastic interventions under two-phase sampling, \citet{hejazi_efficient_2021} study efficient estimators, sampling-mechanism targeting, remainder conditions, and multiple robustness. Our analysis uses this established theory for the standardized contrast in~\eqref{eq:full_data_target}.

Section~\ref{sec:preliminary} introduces the notation, the efficient influence curve (EIC), targeted maximum likelihood estimation (TMLE), and the exact remainder defined in Section~\ref{sec:exact_remainder} to study asymptotic behavior. Section~\ref{sec:review} reviews IPCW-TMLE and its extensions, together with generalized raking (GR) \citep{deville_calibration_1992,lumley_connections_2011}. GR has performed well in simulation studies such as those of \citet{shepherd_multiwave_2023} and \citet{williamson_assessing_2026}. The raking estimator studied here combines calibrated weights with a parametric outcome regression. We call the population limit of the standardized contrast from fitting that working regression to the complete phase-1 cohort the census estimand. This terminology connects our comparison to the survey-sampling goal of recovering a specified complete-cohort analysis \citep{breslow_improved_2009}. Section~\ref{sec:raking} defines this target and distinguishes consistency for it from consistency for $\psi_0$. The census estimand can differ from $\psi_0$, even when the causal assumptions hold. For estimation of $\psi_0$, this particular raking estimator is less robust to outcome-regression misspecification than TMLE when the treatment and sampling mechanisms are consistently estimated and the remaining consistency conditions hold. Raking calibration can also be used within TMLE to solve the empirical equation for the sampling-mechanism component of the EIC\@.

Our contribution is the construction and comparison of specific estimators within this framework. Section~\ref{sec:method_aipcw_rearrange} develops an efficient estimating-equation (EEE) estimator by targeting a conditional expectation of the uncentered full-data EIC, and a Quasi-TMLE that couples this equation to a plug-in constraint. Section~\ref{sec:real_tmle} describes a TMLE with separate updates for conditional expectations in the target parameter. We give computational details for these regressions and updates. Section~\ref{sec:analyze_exact_rem} specializes the established remainder decomposition and gives an explicit bound and a componentwise refinement. Section~\ref{sec:simulation} compares finite-sample performance, and Section~\ref{sec:conclusion} discusses limitations.

\section{Preliminaries}\label{sec:preliminary}
We begin by defining the notation used throughout the article. We then review the efficient influence curve, TMLE, and the exact remainder defined in Section~\ref{sec:exact_remainder}. These established results provide the basis for the estimators reviewed in Section~\ref{sec:review} and the constructions in Sections~\ref{sec:method_aipcw_rearrange} and~\ref{sec:real_tmle}.

We distinguish recovery of the full-data statistical target from its causal interpretation. Throughout, we assume CAR, $\Delta\perp X\mid V$, and sampling positivity, $\Pi_0(V)>0$ almost surely under $P_{0,V}$. These assumptions identify the full-data distribution from the observed-data distribution. Consistency and conditional exchangeability, stated in Section~\ref{sec:introduction}, are needed for the causal equality in~\eqref{eq:causal_identification}. They are not needed for the statistical EIC and remainder calculations below, which retain their own positivity, moment, and regularity conditions.

\subsection{Notation and the observed-data model}\label{sec:notation}
For $P_X\in\mathcal{M}^F$ and a sampling mechanism $\Pi$, let $P_{P_X,\Pi}$ denote the distribution generated by drawing $X\sim P_X$, drawing $\Delta\mid X\sim\Bern\{\Pi(V)\}$, and retaining $O=(V,\Delta,\Delta W_2)$. This construction imposes CAR\@. Let $\mathcal{G}$ be the allowed class of sampling mechanisms. The induced observed-data model is
\begin{equation}\label{eq:observed_data_model}
\mathcal{M}=\{P_{P_X,\Pi}:P_X\in\mathcal{M}^F,\ \Pi\in\mathcal{G},\ 0<\Pi(V)\leq1\ \ P_{X,V}\text{-almost surely}\}.
\end{equation}
The class $\mathcal{G}$ may be unrestricted or may encode design information. If the sampling mechanism is known, $\mathcal{G}$ contains only that mechanism. This defines $\mathcal{M}$ as a family of observed-data distributions induced by the full-data model and sampling design \citep{rose_targeted_2011}.

Conversely, for $P\in\mathcal{M}$, let $P_V$ be its marginal distribution of $V$ and $\Pi_P(v)\coloneq P(\Delta=1\mid V=v)$. Since $X$ is observed when $\Delta=1$, the conditional distribution
$$
\Gamma_P(dx\mid v)\coloneq P(X\in dx\mid\Delta=1,V=v)
$$
is identified on the support of $P_V$ by sampling positivity. Under CAR, it equals the full-data conditional distribution of $X$ given $V=v$. Thus the unique full-data distribution compatible with $P$ under these assumptions is
\begin{equation}\label{eq:recover_full_data}
P_X(P)(dx)=\int\Gamma_P(dx\mid v)\,P_V(dv).
\end{equation}
Here $P_X(P)$ denotes a recovered distribution, not a marginal of the coarsened observation $O$. We define the observed-data statistical target by
\begin{equation}\label{eq:observed_data_target}
\Psi(P)\coloneq\Psi^F\{P_X(P)\},\qquad\Psi(P_{P_X,\Pi})=\Psi^F(P_X).
\end{equation}
In particular, $P_0=P_{P_{X,0},\Pi_0}$ and $\psi_0=\Psi(P_0)=\Psi^F(P_{X,0})$. The additional causal assumptions imply $\psi_0=\theta_0$. Changing $\Pi$ while holding $P_X$ fixed changes the observed-data distribution but leaves the target value unchanged.

Table~\ref{tab:notations} summarizes the notation. To simplify expressions, for an observed-data distribution $P$ we write $Q_P\coloneq Q_{P_X(P)}$, $g_P\coloneq g_{P_X(P)}$, and $D_P^F\coloneq D^F_{P_X(P)}$. Thus $Q_P$ and $g_P$ are computed under the recovered full-data distribution, not under the distribution restricted to phase-2 participants. We use $P_{X\mid V}$ and $\Gamma_P$ for the same conditional distribution, suppressing the subscript on $\Gamma$ when the distribution is clear. Subscript $0$ denotes the true quantity, and subscript $n$ an estimate. We write $Pf\equiv\int f(o)\,dP(o)$ for an observed-data expectation and $P_n f\equiv n^{-1}\sum_{i=1}^{n}f(O_i)$. Expectations of full-data functions are written under $P_X(P)$ explicitly. Conditional expectations given $\Delta=1,V$ are under the observed-data distribution and agree with their full-data counterparts under CAR\@.

We use ``statistical model'' for a collection of probability distributions, ``working regression'' or ``regression class'' for a specified family used in nuisance estimation, and ``learner'' for a fitting algorithm. A fluctuation submodel is a path through an initial fit used for targeting. We write $\widehat P_n$ for an initial estimated distribution and $P_n^\star$ for a targeted estimate, with the estimated components used by each procedure specified locally. These estimates are distinct from the empirical averaging operator $P_n f$. The notation $E_n(f\mid\Delta=1,V)$ denotes a fitted conditional regression, not conditioning the empirical measure. Superscript $(k)$ denotes an algorithmic iterate, while $\star$ denotes a final targeted estimate. In particular, a first iterate need not be the final estimate.

\begin{table}[ht]
\centering
\resizebox{0.8\textwidth}{!}{
\begin{tabular}{rl}
\toprule
\textbf{Notation} & \textbf{Description} \\ \midrule
$W_1\in\RR^d$ & Baseline covariates collected in phase-1 \\
$W_2\in\RR^d$ & Baseline covariates collected in phase-2 \\
$W=(W_1,W_2)$ & Baseline covariates \\
$A\in\{0,1\}$ & Treatment assignment \\
$Y\in\RR$ & Outcome \\ 
$\Delta\in\{0,1\}$ & Phase-2 membership indicator \\ 
$V=(W_1,A,Y)$ & Variables fully observed \\ 
$O=(V,\Delta,\Delta W_2)$ & Observed data \\
$X=(W,A,Y)$ & Uncoarsened data, without both potential outcomes\\
$\mathcal{M}^F$ & Statistical model for distributions of $X$\\
$P_{P_X,\Pi}$ & Observed-data distribution induced by $P_X$ and $\Pi$\\
$\mathcal{M}$ & Induced observed-data model in~\eqref{eq:observed_data_model}\\
$P_X(P)$ & Full-data distribution recovered from $P$ under CAR and positivity\\
$\Gamma_P(dx\mid v)$ & Conditional distribution of $X$ given $\Delta=1,V=v$\\ \midrule
$Q_P=Q_{P_X(P)}$ & Full-data outcome regression\\
$g_P=g_{P_X(P)}$ & Full-data propensity score for treatment $A=1$\\
$\Pi_{P}(V)\coloneq P(\Delta=1\mid V)$ & Phase-2 sampling probability \\ \midrule
$\Psi^F(P_X)$ & Full-data statistical contrast in~\eqref{eq:full_data_target}\\
$\Psi(P)=\Psi^F\{P_X(P)\}$ & Observed-data statistical target\\
$\psi_0=\Psi(P_0)$ & True value of the statistical target\\
$\theta_0=E(Y_1-Y_0)$ & Causal ATE, equal to $\psi_0$ under~\eqref{eq:causal_identification}\\
$D^F_{P_X}(X)$ & Efficient influence curve of $\Psi^F$ at $P_X\in\mathcal{M}^F$\\
$P_n f$ & Empirical average $n^{-1}\sum_{i=1}^n f(O_i)$\\
$\widehat P_n,\ P_n^\star$ & Initial and targeted estimates, respectively\\
\bottomrule
\end{tabular}
}
\caption{Notations and their descriptions.}
\label{tab:notations}
\end{table}

\subsection{The efficient influence curve}\label{sec:eic}
In the full-data world, the efficient influence curve (EIC) $D^F_{P_X}$ for the full-data parameter $\Psi^F$ at $P_X\in\mathcal{M}^F$, where $\mathcal{M}^F$ is the statistical model for the full data, is well known in the literature and is given by
$$
D^F_{P_X}(X)=\left[\frac{A}{g_{P_X}(W)}-\frac{1-A}{1-g_{P_X}(W)}\right](Y-Q_{P_X}(A,W))+Q_{P_X}(1,W)-Q_{P_X}(0,W)-\Psi^F(P_X).
$$
We refer readers to \citet{levy_tutorial_2019} and \citet{hines_demystifying_2022} for tutorials on deriving EICs. Under CAR and sampling positivity, the established full-data to observed-data mapping of \citet{robins_estimation_1994,robins_analysis_1995}, Theorem~1.3 of \citet{vdl_unified_2003}, and Theorem~7.2 of \citet{tsiatis_semiparametric_2006} gives the EIC of $\Psi$ defined in~\eqref{eq:observed_data_target}. This mapping is also used by \citet{kennedy_missing_2020} for missing exposure information. For the nonparametric full-data model, it is the IPCW full-data EIC minus its projection onto the sampling-mechanism score space of the unrestricted CAR model:
$$
D_{P_0}=\underbrace{\frac{\Delta D^F_{P_{X,0}}}{\Pi_0(V)}}_{\text{IPCW}}-\underbrace{\frac{E(D^F_{P_{X,0}}\mid\Delta=1,V)}{\Pi_0(V)}(\Delta-\Pi_0(V))}_{\text{proj. onto $\mathcal{T}$(CAR)}}.
$$
We will also refer to the projection on $\mathcal{T}(\text{CAR})$ term as the $\Pi$-component of the EIC since it is of the form a score in $\Pi$, i.e., a function of $V$ multiplied by the residual $(\Delta-\Pi_0(V))$. 

For later use, define $\bar D_P^F=D_P^F+\Psi(P)$, the uncentered full-data EIC, and its conditional regressions
\begin{equation}\label{eq:conditional_eic_regressions}
\begin{aligned}
m_P(V)&=E_P(D_P^F\mid\Delta=1,V),\\
\bar m_P(V)&=E_P(\bar D_P^F\mid\Delta=1,V)=m_P(V)+\Psi(P).
\end{aligned}
\end{equation}
We use $m_n$ and $\bar m_n$ for fitted regressions of the corresponding centered and uncentered responses, adding iteration or fluctuation subscripts when needed. The population identity in~\eqref{eq:conditional_eic_regressions} does not imply that two separately fitted regressions differ by exactly a constant. That identity holds for fitted regressions when one is constructed by shifting the other. The conditional regressions in Section~\ref{sec:real_tmle} have different responses and are written explicitly. The working-coefficient auxiliary $\widehat h_A$ in Section~\ref{sec:raking} and the conditional-loss regression $m^L_{n,q}$ in Section~\ref{sec:dr_initial_learning} also have different responses and retain distinct notation.

\subsection{Targeted maximum likelihood estimation}
Targeted maximum likelihood estimation (TMLE) (\cite{vdl_targeted_2006, vdl_targeted_2011, vdl_targeted_2018}) is a general framework for constructing asymptotically linear, locally efficient, and \textit{plug-in} estimators of statistical parameters under realistic assumptions on the data-generating process. A key distinguishing feature of TMLE is its plug-in property, which ensures that the resulting estimator respects known bounds of the parameter space (e.g., estimated probabilities remain within $[0,1]$). TMLE is designed to integrate modern machine learning (ML) algorithms via super learner (SL), an ensemble ML approach (\cite{vdl_super_2007}), while still yielding valid statistical inference for a broad class of target estimands often of interest in causal inference literature. TMLE builds on a long line of work in semiparametric statistics, including one-step estimators (\cite{bickel_efficient_1993}), estimating equations (augmented inverse probability of censoring weighting, i.e., A-IPCW) (\cite{robins_estimation_1994,robins_analysis_1995,vdl_unified_2003, tsiatis_semiparametric_2006}), and sieve MLE (\cite{shen_on_1997,chen_large_2007}). Each of these methods contributed important ideas but had limitations. For example, one-step and A-IPCW estimators, while efficient, can be unstable due to not being plug-in.

TMLE addresses these issues by bringing the likelihood principle back into focus, while also tailoring the estimation to the particular estimand of interest. This results in an estimator that retains the desirable theoretical properties of earlier approaches, such as double robustness, while offering robust finite-sample performance. The TMLE procedure generally begins by estimating key components of the likelihood. It then performs a targeting step, which updates the initial estimate along a fluctuation submodel through the initial likelihood, with the goal of solving the empirical mean of the efficient influence curve evaluated at the targeted estimate. As we will show in the following subsection, solving this empirical score equation is a key step for ensuring asymptotically valid statistical inference and constructing Wald-type confidence intervals.

TMLE has been developed as a general estimation framework for a wide range of data structures and target parameters, including point-treatment settings, time-to-event outcomes (\cite{moore_increasing_2009,cai_onestep_2020,rytgaard_continuous_2022,chen_beyond_2023}), longitudinal data (\cite{vdl_ltmle_2011,shirakawa_longitudinal_2024}), and network data (\cite{vdl_causal_2012,ogburn_causal_2024}) (also see other examples in the two books on TMLE, \cite{vdl_targeted_2011} and \cite{vdl_targeted_2018}). For two-phase sampling designs, an IPCW-TMLE has been proposed (\cite{rose_targeted_2011}), and will be reviewed in detail in Section \ref{sec:ipcw_tmle}.

\subsection{The exact remainder term}\label{sec:exact_remainder}
For a pathwise differentiable parameter $\Psi:\mathcal{M}\rightarrow\RR$, define the exact remainder
\begin{equation}\label{eq:exact_remainder_definition}
R(P,P_0)\coloneq\Psi(P)-\Psi(P_0)+P_0D_P.
\end{equation}
We use $R$ for this quantity and reserve $R_n$ below for the total remainder in the estimator's expansion. Let $P_n^\star$ be the targeted estimate defined in Section~\ref{sec:notation}. If its plug-in estimator is asymptotically linear and efficient, then
$$
\Psi(P_n^\star)-\Psi(P_0)=P_nD_{P_0}+o_P(n^{-1/2}).
$$
Define the total remainder $R_n(P_n^\star,P_0)\coloneq\Psi(P_n^\star)-\Psi(P_0)-P_nD_{P_0}$. Note that we have
\begin{align*}
R_n(P_n^\star,P_0)&=\Psi(P_n^\star)-\Psi(P_0)+P_0D_{P_n^\star}+(P_n-P_0)\{D_{P_n^\star}-D_{P_0}\}-P_nD_{P_n^\star}\\
&=\underbrace{R(P_n^\star,P_0)}_{\text{exact remainder}}+\underbrace{(P_n-P_0)\{D_{P_n^\star}-D_{P_0}\}}_{\text{empirical process}}-\underbrace{P_n D_{P_n^\star}}_{\text{score residual}}.
\end{align*}
Suppose $D_{P_n^\star}$ belongs to a Donsker class with probability tending to one and $P_0{\{D_{P_n^\star}-D_{P_0}\}}^2\convp0$. Asymptotic equicontinuity then gives $(P_n-P_0)\{D_{P_n^\star}-D_{P_0}\}=o_P(n^{-1/2})$ \citep{vaart_weak_1997}. This argument requires both the class restriction and influence-curve consistency. If the TMLE satisfies $P_n D_{P_n^\star}=o_P(n^{-1/2})$, its numerical score residual is also negligible \citep{vdl_targeted_2006}. The remaining requirement in this expansion is $R(P_n^\star,P_0)=o_P(n^{-1/2})$.
\begin{remark}\label{rem:cv_tmle}
The Donsker condition controls the empirical-process term when the same observations are used to fit and evaluate a data-adaptive influence curve. Cross-validated TMLE (CV-TMLE) fits initial nuisance functions on training folds and evaluates their predictions on held-out validation folds. Conditional on the training data, those initial fits are fixed relative to the validation observations. This separation permits conditional empirical-process arguments for the initial fits, together with control of the low-dimensional targeting step, without a Donsker restriction on the initial learners \citep{zheng_asymptotic_2010}. The required influence-curve consistency, moment, targeting, and second-order remainder conditions still apply. Cross-validation used only to select a nuisance learner, followed by refitting and evaluating it on the full sample, does not itself provide this separation.
\end{remark}

The exact remainder term often involves second-order differences between the estimated and true components of the data-generating distribution. For example, for the full-data statistical contrast in~\eqref{eq:full_data_target}, the exact remainder takes the form of an integral involving the product of two approximation errors: one for the outcome regression and the other for the treatment mechanism. To ensure that the exact remainder vanishes at the desired rate, each component must converge to its true value at a rate faster than $n^{-1/4}$, or any product of rates that would result in $R(P_n^\star,P_0)$ being $o_P(n^{-1/2})$. This product structure makes the full-data remainder vanish when one nuisance function is consistent in $L^2$, the other has bounded $L^2$ error, and the estimated treatment probabilities remain uniformly bounded away from zero and one. Consistency of the resulting estimator also requires control of the other terms in its expansion. The exact remainder also motivates the use of super learner to data-adaptively select among a library of candidate estimators, rather than relying on a single algorithm. In particular, including learners with weak restrictions on the regression function yet provably fast convergence rates, such as the highly adaptive lasso (\cite{vdl_generally_2017,benkeser_hal_2016}), can endow the super learner with favorable theoretical guarantees. Importantly, super learner is not restricted to machine learning algorithms: in settings where partial knowledge of the data-generating process is available (e.g., a known treatment mechanism), one may include both parametric regression fits and nonparametric learners in the library to reflect such knowledge.

\section{Review of existing estimators for two-phase sampling designs}\label{sec:review}
Methods for analyzing two-phase data include multiple imputation \citep{rubin_multiple_1978}, augmented inverse probability of censoring weighting (A-IPCW) \citep{vdl_unified_2003}, generalized raking \citep{deville_calibration_1992,breslow_improved_2009,breslow_using_2009}, and IPCW-TMLE \citep{rose_targeted_2011,hejazi_efficient_2021}. These belong to the broader literature on efficient two-phase estimation, double sampling, and coarsened-data inference discussed in Section~\ref{sec:introduction} \citep{gilbert_optimal_2014,levis_double_2024,kennedy_missing_2020}. Our detailed comparison focuses on IPCW-TMLE variants and generalized raking for the missing-confounder setting studied by \citet{williamson_assessing_2026}. This selection is not an exhaustive review of efficient estimators for all two-phase designs.

The comparisons in \citet{williamson_assessing_2026} included synthetic data and simulations based on an electronic health record cohort, with variation in regression complexity, outcome rarity, and the amount of missing confounder data. Generalized raking often had favorable efficiency for census targets. TMLE had coverage advantages in some complex settings, including rare-outcome scenarios, although its variance could be larger. These findings motivate including both approaches here. Their relative performance depended on the estimand and data-generating process, so they do not establish a general ranking for two-phase studies.

\subsection{Inverse probability of censoring weighted targeted maximum likelihood estimator (IPCW-TMLE)}\label{sec:ipcw_tmle}
An implementation of IPCW-TMLE is available in the \texttt{R} package \texttt{twoStageDesignTMLE} \citep{gruber_twostage_2024}. The initial fits described here use ordinary IPCW learning. Section~\ref{sec:dr_initial_learning} discusses an augmented-risk construction for initial nuisance learning, distinct from targeting the final treatment-effect influence curve.

The IPCW-TMLE estimator proposed in~\cite{rose_targeted_2011} works as follows. First, obtain an estimator $\Pi_n$ of the phase-2 sampling mechanism $\Pi_0$. Next, estimate the full-data outcome regression $Q_0$ and treatment mechanism $g_0$ by fitting initial estimators $Q_n$ and $g_n$, respectively, using appropriate full-data loss functions weighted by $\Delta/\Pi_n$. Then, the initial estimator $Q_n$ is targeted by fitting a univariate logistic regression with outcome $Y$, offset $\text{logit}\ Q_n$, clever covariate $H_n(A,W)=A/{g_n(W)}-(1-A)/\{1-g_n(W)\}$, and weights $\Delta/\Pi_n$. Note that for bounded continuous outcome $Y$, one may scale it to $(0,1)$. This targeting step yields the updated estimator $Q_n^\star$, which solves the equation $P_n(\Delta D^F_{P_n^\star}/\Pi_n)=0$ at $P_n^\star=(Q_n^\star,\Pi_n,g_n)$, where
$$
D^F_{P_n^\star}\coloneq H_n(A,W)(Y-Q_n^\star(A,W))+Q_n^\star(1,W)-Q_n^\star(0,W)-\Psi(Q_n^\star,P_{W,n})
$$
is the full-data EIC evaluated at the estimated nuisance functions. Consistent with the implementation, $P_{W,n}$ is the normalized weighted empirical distribution of $W$ among phase-2 participants, assigning mass
$$
\omega_{i,n}=\frac{\Delta_i/\Pi_n(V_i)}{\sum_{j=1}^n\Delta_j/\Pi_n(V_j)}
$$
to observation $i$. This normalization makes $P_{W,n}$ a probability distribution and makes the weighted empirical mean of the centered contrast zero. Algorithm~\ref{alg:ipcw_tmle} summarizes the estimator. Throughout the algorithms, full-data expressions multiplied by $\Delta$ are evaluated only for phase-2 participants and contribute zero otherwise. Sampling regressions and empirical score averages use the full cohort.

\begin{algorithm}
\caption{IPCW-TMLE estimator}
\begin{algorithmic}[1]\label{alg:ipcw_tmle}
\STATE{} Estimate $\Pi_n$ using phase-1 and phase-2 data;
\STATE{} Estimate $Q_n$ and $g_n$ using phase-2 data with weights $\Delta/\Pi_n$;
\STATE{} On phase-2 data, fit a univariate logistic regression with outcome $Y$, offset $\text{logit }Q_n(A,W)$, covariate
$$
H_n(A,W) = \frac{A}{g_n(W)} - \frac{1 - A}{1-g_n(W)},
$$
and weight $\Delta/\Pi_n(V)$;
\STATE{} Update $Q_n$, keeping $\Pi_n$ and $g_n$ fixed:
$$
\text{logit }Q_n^\star(A,W) \coloneq \text{logit }Q_n(A,W) + \epsilon_n H_n(A,W),
$$
where $\epsilon_n$ is the estimated coefficient from the univariate logistic regression;
\STATE{} Compute the normalized weighted contrast over phase-2 participants:
$$
\Psi(P_n^\star)=\frac{\sum_{i=1}^n\Delta_i\{Q_n^\star(1,W_i)-Q_n^\star(0,W_i)\}/\Pi_n(V_i)}{\sum_{i=1}^n\Delta_i/\Pi_n(V_i)}.
$$
\end{algorithmic}
\end{algorithm}

Knowing the sampling mechanism does not generally make the augmentation term in Section~\ref{sec:eic} zero. For the nonparametric full-data model under CAR, the EIC retains this term even when $\Pi_0$ is fixed by design. Updating sampling weights can therefore improve efficiency even when their initial values are known. Such an update is an estimation device and does not alter the actual sampling design. When $\Pi_0$ is estimated, the sampling-model score equations may already supply part or all of this augmentation \citep{rose_targeted_2011,vdl_unified_2003}. Algorithm~\ref{alg:ipcw_tmle} targets only $Q$, leaving $\Pi_n$ and $g_n$ fixed. Algorithm~\ref{alg:ipcw_tmle_target_Pi} additionally targets $\Pi$. For an iterative procedure, $Q_n^\star=Q_n^{(K)}$ and $\Pi_n^\star=\Pi_n^{(K)}$ denote the final iterates when that notation is used; $\Pi_n^{(1)}$ denotes only the first sampling update. The next subsection describes explicit targeting and the conditions relevant to its benefit.

\subsection{IPCW-TMLE with targeted phase-2 sampling mechanism}\label{sec:ipcw_tmle_target_Pi_local}
\citet{rose_targeted_2011} discussed targeting the phase-2 sampling mechanism. For stochastic interventions under two-phase sampling, \citet{hejazi_efficient_2021} developed and evaluated such targeting; their Section~3.2.2 and Algorithm~1 in Supplementary Section~S2 describe the updates. Their methods are implemented in the \texttt{R} package \texttt{txshift} \citep{hejazi_txshift_2020}. We use this targeting strategy for the binary-treatment standardized contrast and include it as a comparator in our simulations. The full-data estimand, and hence its EIC, differs from the stochastic-intervention estimand in that work. We first examine the efficiency rationale for targeting the sampling mechanism, then describe the updates used here.

To quantify the potential efficiency gain, we consider a benchmark using the true sampling mechanism and full-data EIC\@.

\begin{lemma}[Variance reduction from augmentation]\label{lem:target_pi_variance_gain}
Suppose CAR holds, $0<\Pi_0(V)\leq 1$ almost surely, and the full-data EIC satisfies $E_0[{(D_0^F)}^2/\Pi_0(V)]<\infty$. Define
$$
m_0(V)=E_0(D_0^F\mid\Delta=1,V),\qquad U_0=\frac{\Delta}{\Pi_0(V)}D_0^F,\qquad D_0=U_0-\left\{\frac{\Delta}{\Pi_0(V)}-1\right\}m_0(V).
$$
Then $U_0$ and $D_0$ have mean zero and finite variance, $E_0[D_0(U_0-D_0)]=0$, and
\begin{equation}\label{eq:target_pi_variance_gain}
\Var_0(U_0)-\Var_0(D_0)=E_0\left[\frac{1-\Pi_0(V)}{\Pi_0(V)}{m_0(V)}^2\right].
\end{equation}
\end{lemma}

This is the projection-based variance comparison underlying efficient two-phase estimation \citep{rose_targeted_2011,vdl_unified_2003}; its proof is given in Appendix~\ref{sec:proof_target_pi_variance_gain}. For regular estimators with influence functions $U_0$ and $D_0$, the leading variance difference is the right-hand side of~\eqref{eq:target_pi_variance_gain} divided by $n$. The variance gain depends on the magnitude of the conditional mean $m_0(V)$ and the phase-2 sampling probabilities. It can be large when $m_0(V)$ is far from zero in strata with small sampling probabilities, even if individual full-data EIC values vary substantially around their conditional mean. Augmentation uses information available on everyone to account for this predictable component, reducing the variability that arises from estimating it through the weighted phase-2 sample. The gain is zero when $m_0(V)=0$ wherever $\Pi_0(V)<1$. In that case, the augmentation term vanishes because the conditional mean is zero in incompletely sampled strata, while fully sampled strata require no correction.

The accuracy of the conditional EIC regression also matters. The next lemma quantifies the variance penalty for using an approximation to $m_0(V)$.

\begin{lemma}[Variance with approximate augmentation]\label{lem:target_pi_approx_variance}
Under the assumptions of Lemma~\ref{lem:target_pi_variance_gain}, let $h(V)$ be a fixed measurable function satisfying $E_0[{h(V)}^2/\Pi_0(V)]<\infty$, and define
$$
D_h=U_0-\left\{\frac{\Delta}{\Pi_0(V)}-1\right\}h(V).
$$
Then $D_h$ has mean zero and finite variance, and
$$
\Var_0(D_h)=\Var_0(D_0)+E_0\left[\frac{1-\Pi_0(V)}{\Pi_0(V)}{\{h(V)-m_0(V)\}}^2\right].
$$
Consequently, $\Var_0(D_h)<\Var_0(U_0)$ if and only if
$$
E_0\left[\frac{1-\Pi_0(V)}{\Pi_0(V)}{\{h(V)-m_0(V)\}}^2\right]<E_0\left[\frac{1-\Pi_0(V)}{\Pi_0(V)}{m_0(V)}^2\right].
$$
\end{lemma}

The proof is given in Appendix~\ref{sec:proof_target_pi_approx_variance}. The additional variance is the weighted squared error in approximating $m_0(V)$, with greater weight on errors in strata with small sampling probabilities. Since $h(V)=0$ gives $D_h=U_0$, augmentation reduces variance whenever $h(V)$ approximates $m_0(V)$ better than the zero function under this criterion. Thus misspecification of the conditional EIC regression alone does not imply an absence of improvement.

For example, take $h(V)=\tfrac12m_0(V)$. This approximation recovers only half of the conditional mean, so ${\{h(V)-m_0(V)\}}^2=\tfrac14{m_0(V)}^2$. Lemmas~\ref{lem:target_pi_variance_gain} and~\ref{lem:target_pi_approx_variance} then give
$$
\Var_0(U_0)-\Var_0(D_h)=\frac34\{\Var_0(U_0)-\Var_0(D_0)\}.
$$
Thus this imperfect approximation achieves three-quarters of the variance reduction available from the ideal augmentation. These comparisons are benchmarks, rather than exact variance formulas for every fitted TMLE\@. With estimated $\Pi$, the untargeted estimator may already recover the relevant projection. Small phase-2 samples and unstable inverse weights can make the additional regression and targeting steps difficult, so a first-order efficiency argument does not guarantee lower finite-sample MSE\@.

We now describe the alternating updates in Algorithm~\ref{alg:ipcw_tmle_target_Pi}. Each cycle first targets $Q$ using the current sampling weights, then targets $\Pi$ using the fitted conditional expectation of the updated full-data EIC\@. The treatment mechanism $g_n$ remains fixed. In the algorithm boxes, $P_n^{(k)}$ denotes the current fitted components used to evaluate the score, including any fitted conditional regressions, while $P_n f$ denotes empirical averaging. Conditional regressions are fitted among phase-2 participants and predicted for everyone. At each convergence check, the score is recomputed using the current $Q$, $\Pi$, and conditional regression. The full-data EIC is centered at $\Psi(Q_n^{(k)},P_{W,n}^{(k)})$, where $P_{W,n}^{(k)}$ uses the normalized weights based on $\Pi_n^{(k)}$ as in Section~\ref{sec:ipcw_tmle}. Appendix~\ref{sec:targeting_implementation} describes the numerical implementation.

Algorithm~\ref{alg:ipcw_tmle_target_Pi} uses the adaptive tolerance
\begin{equation}\label{eq:adaptive_tolerance}
s_n^{(k)}=\frac{\sigma_n^{(k)}}{\sqrt n\log n},
\end{equation}
where $\sigma_n^{(k)}$ is the empirical standard deviation of the current EIC values $D_{P_n^{(k)}}(O_i)$ over all $n$ observations. It is a standard deviation, and $\sigma_n^{(k)}/\sqrt n$ is the estimated standard error of the score mean. Both the score and the tolerance are recomputed after each update cycle. If $|P_n D_{P_n^{(K)}}|\leq s_n^{(K)}$ at the final iteration $K$ and $\sigma_n^{(K)}=O_P(1)$, then
$$
\sqrt n\,|P_n D_{P_n^{(K)}}|\leq\frac{\sigma_n^{(K)}}{\log n}=o_P(1).
$$
This makes the numerical score residual negligible at the scale of first-order inference.

In a finite sample, an increase in $\sigma_n^{(k)}$ increases the tolerance and can satisfy the criterion without reducing the absolute score residual. For example, at $n=1000$, a residual of $0.001$ fails the criterion when $\sigma_n^{(k)}=0.1$ but passes when it rises to $0.3$. The adaptive rule is therefore a relative score criterion, not a guarantee of monotone numerical improvement. If increases in the adaptive tolerance are a concern, an alternative is to compute a reference tolerance from the initial score standard deviation and hold it fixed throughout subsequent iterations. This prevents an increase in the score standard deviation from relaxing the stopping criterion. Provided the initial score standard deviation is $O_P(1)$, satisfying this fixed criterion also makes the numerical score residual $o_P(n^{-1/2})$. The algorithms studied here use the adaptive tolerance in~\eqref{eq:adaptive_tolerance}.

\begin{algorithm}
\caption{IPCW-TMLE estimator with targeted phase-2 sampling mechanism}
\begin{algorithmic}[1]\label{alg:ipcw_tmle_target_Pi}
\STATE{} Obtain initial estimator $\Pi_n^{(0)}$ for $\Pi_0$ using phase-1 and phase-2 data;
\STATE{} Obtain initial estimators $Q_n^{(0)}$ and $g_n$ using phase-2 data with weights $\Delta/\Pi_n^{(0)}$;
\STATE{} Compute the initial EIC and define $s_n^{(0)}=\sigma_n^{(0)}/(\sqrt n\log n)$, where $\sigma_n^{(0)}$ is its empirical standard deviation over everyone;
\STATE{} Set iteration counter $k=0$;
\WHILE{$|P_n D_{P_n^{(k)}}|>s_n^{(k)}$}
\STATE{} On phase-2 data, fit a univariate logistic regression with outcome $Y$, offset $\operatorname{logit}Q_n^{(k)}(A,W)$, covariate
$$
H_n(A,W)=\frac{A}{g_n(W)}-\frac{1-A}{1-g_n(W)},
$$
and weights $\Delta/\Pi_n^{(k)}(V)$;
\STATE{} With fitted coefficient $\epsilon_n^{(k)}$, update $Q_n^{(k)}$:
$$
\operatorname{logit}Q_n^{(k+1)}(A,W)=\operatorname{logit}Q_n^{(k)}(A,W)+\epsilon_n^{(k)}H_n(A,W);
$$
\STATE{} Using $Q_n^{(k+1)}$ and the current sampling weights, compute $D_n^{F,(k+1)}$ and fit $E_n(D_n^{F,(k+1)}\mid\Delta=1,V)$ on phase-2 data, predicting for everyone;
\STATE{} Using everyone, fit a univariate logistic regression with outcome $\Delta$, offset $\operatorname{logit}\Pi_n^{(k)}(V)$, and covariate
$$
C_n^{(k)}(V)=\frac{E_n(D_n^{F,(k+1)}\mid\Delta=1,V)}{\Pi_n^{(k)}(V)};
$$
\STATE{} With fitted coefficient $\delta_n^{(k)}$, update $\Pi_n^{(k)}$:
$$
\operatorname{logit}\Pi_n^{(k+1)}(V)=\operatorname{logit}\Pi_n^{(k)}(V)+\delta_n^{(k)}C_n^{(k)}(V);
$$
\STATE{} Recompute $D_{P_n^{(k+1)}}$ and its empirical standard deviation $\sigma_n^{(k+1)}$ over everyone, including recentering the full-data EIC and updating its conditional regression. Set $s_n^{(k+1)}=\sigma_n^{(k+1)}/(\sqrt n\log n)$ as in~\eqref{eq:adaptive_tolerance};
\STATE{} Increase $k$ by one;
\ENDWHILE{}
\STATE{} Set $K=k$ and compute the final normalized weighted contrast:
$$
\Psi(P_n^\star)=\frac{\sum_{i=1}^n\Delta_i\{Q_n^{(K)}(1,W_i)-Q_n^{(K)}(0,W_i)\}/\Pi_n^{(K)}(V_i)}{\sum_{i=1}^n\Delta_i/\Pi_n^{(K)}(V_i)}.
$$
\end{algorithmic}
\end{algorithm}

\begin{remark}\label{rem:linearization}
Refitting the conditional EIC regression at each outer iteration can be computationally expensive. The following linearization reduces repeated regression fitting. When a linear fluctuation submodel for $Q$ is used, the dependence on the fluctuation coefficient is exact. Specifically, under the linear fluctuation $Q_{n,\epsilon}(A,W) = Q_n(A,W)+\epsilon H_n(A,W)$, we can express the un-centered full-data EIC at $Q_{n,\epsilon}$ as $\bar{D}^F_{n,\epsilon} = \bar{D}^F_n + \epsilon G_n$, where $G_n = H_n(1, W) - H_n(0, W) - H_n^2(A,W)$. That is, the un-centered full-data EIC is linear in $\epsilon$ under the linear fluctuation of $Q$. Although we generally recommend the use of logistic fluctuation submodels for TMLE targeting, in this case, a linear submodel is helpful from a computational perspective because it allows us to apply the linearity of expectation to compute $E_n(D^F_{n,\epsilon} \mid \Delta = 1, V)$ without re-fitting the regression at each value of $\epsilon$. In particular, using $\bar m_n$ from~\eqref{eq:conditional_eic_regressions}, construct
$$
m_{n,\epsilon}(V)=\bar m_n(V)+\epsilon E_n(G_n\mid\Delta=1,V)-\Psi(Q_{n,\epsilon}).
$$
This means that the two regressions $E(\bar{D}^F \mid \Delta = 1, V)$ and $E(G \mid \Delta = 1, V)$ only need to be computed once at the initial fit.

In the case of bounded continuous or binary outcomes, where one typically uses a logistic fluctuation submodel to respect the bounds on the outcome, the exact linearity of $\bar{D}^F_{n,\epsilon}$ in $\epsilon$ no longer holds. However, we may consider a local linearization. Rather than using the exact $\bar{D}^F_{n,\epsilon}$ implied by $Q_{n,\epsilon}$, we approximate it via a first-order Taylor expansion around $\epsilon = 0$. Define $J_P(A, W) = Q_P(A, W)(1 - Q_P(A, W))$, and note that the directional derivative of $Q_{n,\epsilon}$ at $\epsilon = 0$ is given by
$$
\dot{Q}_{n,\epsilon}(A, W) = \frac{\partial Q_{n,\epsilon}(A, W)}{\partial \epsilon} \bigg|_{\epsilon = 0} = J_n(A, W) H_n(A, W).
$$
Using this, the corresponding derivative of the un-centered full-data EIC is
$$
\dot{\bar{D}}^F_{n,\epsilon}(X) = J_n(1, W) H_n(1, W) - J_n(0, W) H_n(0, W) - J_n(A, W) H_n^2(A, W),
$$
at $\epsilon=0$. Thus, we obtain the approximation $\bar{D}^F_{n,\epsilon}\approx\bar{D}^F_n + \epsilon \dot{\bar{D}}^F_n$ and define the locally linearized approximation as $\bar{D}^{F,\text{local}}_{n,\epsilon} = \bar{D}^F_n + \epsilon \dot{\bar{D}}^F_n$. With this approximation, construct the corresponding centered regression as
$$
m_{n,\epsilon}^{\mathrm{local}}(V)=\bar m_n(V)+\epsilon E_n(\dot{\bar{D}}^F_n\mid\Delta=1,V)-\Psi(Q_{n,\epsilon}).
$$
As before, the two regression functions only need to be estimated once at the initial fit. The measured gains in Table~\ref{tab:linearization} are modest. Average run time decreases from $7.64$ to $7.16$ seconds at $n=500$, a reduction of approximately $6.3\%$, and from $18.74$ to $18.07$ seconds at $n=2000$, approximately $3.6\%$. The extra derivative regression has an initial fitting cost. Greater savings are plausible when individual conditional-regression fits are expensive or when many refits are avoided. Appendix~\ref{sec:appendix_local_linearization} reports the timing comparison and an example in which extreme estimated treatment probabilities impair the approximation's accuracy.
\end{remark}

\subsection{Generalized raking (GR)}\label{sec:raking}
Generalized raking calibrates inverse sampling weights to reproduce phase-1 totals of auxiliary variables \citep{deville_calibration_1992,breslow_improved_2009,breslow_using_2009}. We first specify the target of the particular raking estimator studied here. In the survey-sampling formulation, a complete-cohort analysis provides a reference that a two-phase weighted analysis seeks to recover. Its regression parameters remain well defined under working-regression misspecification \citep{breslow_improved_2009,lumley_connections_2011}.

Let $\mathcal{Q}_{\mathrm{work}}=\{Q_\beta:\beta\in\mathcal{B}\}$ be a prespecified outcome-regression class, with population fit
$$
\beta^\star(P_X)\coloneq\arg\min_{\beta\in\mathcal{B}}E_{P_X}\left[L\{Y,Q_\beta(A,W)\}\right],\qquad Q^{\mathrm{work}}_{P_X}\coloneq Q_{\beta^\star(P_X)}.
$$
Assume this minimizer exists and is unique. For the continuous-outcome example below, we use a main-term linear regression and squared-error loss, $L(y,q)={(y-q)}^2$. For the binary-outcome simulations in Section~\ref{sec:simulation}, we use a main-term logistic regression and Bernoulli negative log-likelihood loss, $L(y,q)=-y\log q-(1-y)\log(1-q)$. These choices define the population fit even if $Q_{P_X}$ is outside the working class. The corresponding population census estimand is
\begin{equation}\label{eq:census_target}
\Psi^F_{\mathrm{census}}(P_X)\coloneq\int\{Q^{\mathrm{work}}_{P_X}(1,w)-Q^{\mathrm{work}}_{P_X}(0,w)\}\,dP_{X,W}(w).
\end{equation}
Under CAR and sampling positivity, its observed-data representation is $\Psi_{\mathrm{census}}(P)\coloneq\Psi^F_{\mathrm{census}}\{P_X(P)\}$. Unlike $\Psi^F$ in~\eqref{eq:full_data_target}, this functional depends on the chosen regression class and fitting loss.

To distinguish this population target from a finite-cohort quantity, let $\widehat\beta_{\mathrm{cens}}$ minimize the unweighted empirical loss if $X$ were observed for all $n$ phase-1 individuals. Write $\tau_\beta(w)\coloneq Q_\beta(1,w)-Q_\beta(0,w)$. The hypothetical complete-cohort estimate is
$$
\widehat\psi_{\mathrm{cens}}\coloneq\frac{1}{n}\sum_{i=1}^n\tau_{\widehat\beta_{\mathrm{cens}}}(W_i).
$$
Under the usual consistency conditions for the regression fit and empirical standardization, this random finite-cohort quantity converges to~\eqref{eq:census_target}. We use ``census estimand'' for that population limit throughout the paper.

For auxiliary variables $h_n(V)$ available for everyone, raking solves the convex program
$$
\begin{aligned}
&\min_{a_i>0}\sum_{i=1}^n\Delta_i\,d\left(\frac{a_i}{\Pi_n(V_i)},\frac{1}{\Pi_n(V_i)}\right),\\
&\text{subject to }\sum_{i=1}^n h_n(V_i)=\sum_{i=1}^n\frac{\Delta_i a_i}{\Pi_n(V_i)}h_n(V_i).
\end{aligned}
$$
The distance $d(p,q)=p\log(p/q)-p+q$ gives an exponential tilt of the initial weights \citep{chen_optimal_2022}. The auxiliary variables include an intercept, so the calibrated weights $\omega_i\coloneq\Delta_i a_i/\Pi_n(V_i)$ sum to $n$ when calibration is exact.

The implementation adapted from \citet{williamson_assessing_2026} constructs the remaining auxiliary variable using multiple imputation. Let $\phi_A^{\mathrm{work}}(X)$ denote the influence function for the treatment coefficient in the working regression fit. For each imputed dataset, the procedure fits that regression and computes estimated coefficient influence contributions. It averages those contributions over imputations conditional on the available variables $V$ to obtain an auxiliary variable $\widehat h_{A,i}$. Thus $h_n(V_i)={(1,\widehat h_{A,i})}^T$ in our implementation. This is a working approximation to a conditional mean of a coefficient influence function. It is distinct from $E(D^F\mid V)$, where $D^F$ is the nonparametric EIC for $\Psi^F$, and need not estimate that quantity. Appendix~\ref{sec:appendix_raking} gives implementation details.

The final regression minimizes $\sum_{i:\Delta_i=1}\omega_i L\{Y_i,Q_\beta(A_i,W_i)\}$ over $\beta$, giving $\widehat\beta_{\mathrm{GR}}$. The reported estimate uses normalized weighted standardization,
\begin{equation}\label{eq:raking_estimate}
\widehat\psi_{\mathrm{GR}}\coloneq\frac{\sum_{i:\Delta_i=1}\omega_i\tau_{\widehat\beta_{\mathrm{GR}}}(W_i)}{\sum_{i:\Delta_i=1}\omega_i}.
\end{equation}
Consequently, the final outcome regression defines the census contrast, while the imputation regressions and coefficient influence contributions determine the calibration auxiliaries. Under correct sampling weights and regular calibration, misspecification of these auxiliaries need not change the probability limit, although their choice affects efficiency \citep{lumley_connections_2011}.

\begin{samepage}
The regression-based raking estimator considered here seeks to recover the corresponding complete-cohort analysis \citep{breslow_improved_2009,lumley_connections_2011}. Under CAR, sampling positivity, correct sampling weights, and standard regularity conditions, it is consistent for the census estimand $\Psi_{\mathrm{census}}(P_0)$. This estimand remains well defined when the working outcome regression is misspecified, but it can then differ from the target parameter of interest, $\psi_0$. Thus, consistency for the census estimand does not by itself imply consistency for $\psi_0$.

\end{samepage}

Relative to $\psi_0$, the estimation error decomposes as
$$
\widehat\psi_{\mathrm{GR}}-\psi_0=\{\widehat\psi_{\mathrm{GR}}-\Psi_{\mathrm{census}}(P_0)\}+\{\Psi_{\mathrm{census}}(P_0)-\psi_0\}.
$$
The first term vanishes when the estimator is consistent for the census estimand. The second need not vanish. Correct specification, $Q^{\mathrm{work}}_{P_X}=Q_{P_X}$, is sufficient for the census estimand to equal $\Psi^F(P_X)$, but is not necessary because errors can cancel in the integrated contrast. In comparison, the TMLE estimators can remain consistent for $\psi_0$ with a misspecified initial outcome regression when $g_0$ and $\Pi_0$ are consistently estimated and the remaining conditions in Section~\ref{sec:analyze_exact_rem} hold. In this sense, the raking estimator studied here is less robust to misspecification of its final outcome regression for estimation of $\psi_0$. This comparison concerns the specified regression-based estimator, rather than raking calibration in general.

\begin{remark}\label{rem:raking}
As the sample size increases, a consistent raking estimator approaches the census estimand, $\Psi_{\mathrm{census}}(P_0)$. If confidence intervals centered on this estimator shrink to zero in width, they concentrate around that value. When the census estimand differs from the target of interest, $\psi_0$, the probability that these intervals contain $\psi_0$ therefore tends to zero. This can happen even if the intervals provide valid coverage for the census estimand.

To illustrate this distinction, we conducted a simulation study that varies the gap between $\psi_0$ and $\Psi_{\mathrm{census}}(P_0)$. Treatment and outcome are generated from $W$ with independent errors, so consistency and conditional exchangeability give $\psi_0=\theta_0$. We generated four baseline covariates $W=(W_1,W_2,W_3,W_4)$ independently from $N(1,1)$, a binary treatment $A$, and a continuous outcome $Y$. Treatment assignment followed $P(A=1\mid W)=\text{expit}\,(-0.2W_1-0.6W_2+0.2W_4)$, and $Y=Q(A,W)+U_Y$, where $U_Y$ is mean-zero noise and
$$
\begin{aligned}
Q(A,W)&=-0.3+0.4W_1-0.4W_2+0.2W_3-0.1W_4\\
&\quad+\gamma A\left\{2.5I(W_2>1)-2.5I(W_2<0)+2\sin(W_1)\right\}.
\end{aligned}
$$
Here $\gamma\in\{0,0.2,0.4,0.6,0.8,1\}$ controls treatment effect heterogeneity. The raking procedure fits $Q_\beta(a,w)=\beta_0+\beta_A a+\sum_{j=1}^4\beta_j w_j$ by weighted least squares, so its census contrast equals the treatment coefficient in the population least-squares projection. This working regression is correct at $\gamma=0$ and misspecified at the positive values considered. The phase-2 sampling probability $P(\Delta=1\mid V)=\text{expit}\,(0.5W_1)$ satisfies CAR\@. Figure~\ref{fig:raking_cover} evaluates coverage of $\psi_0$, while Table~\ref{tab:raking_census_estimand} evaluates the same raking procedure against its census target. Together, they illustrate that poor coverage of $\psi_0$ can coexist with good coverage of the census estimand.
\begin{figure}[ht]
\centering
\includegraphics[width=\linewidth]{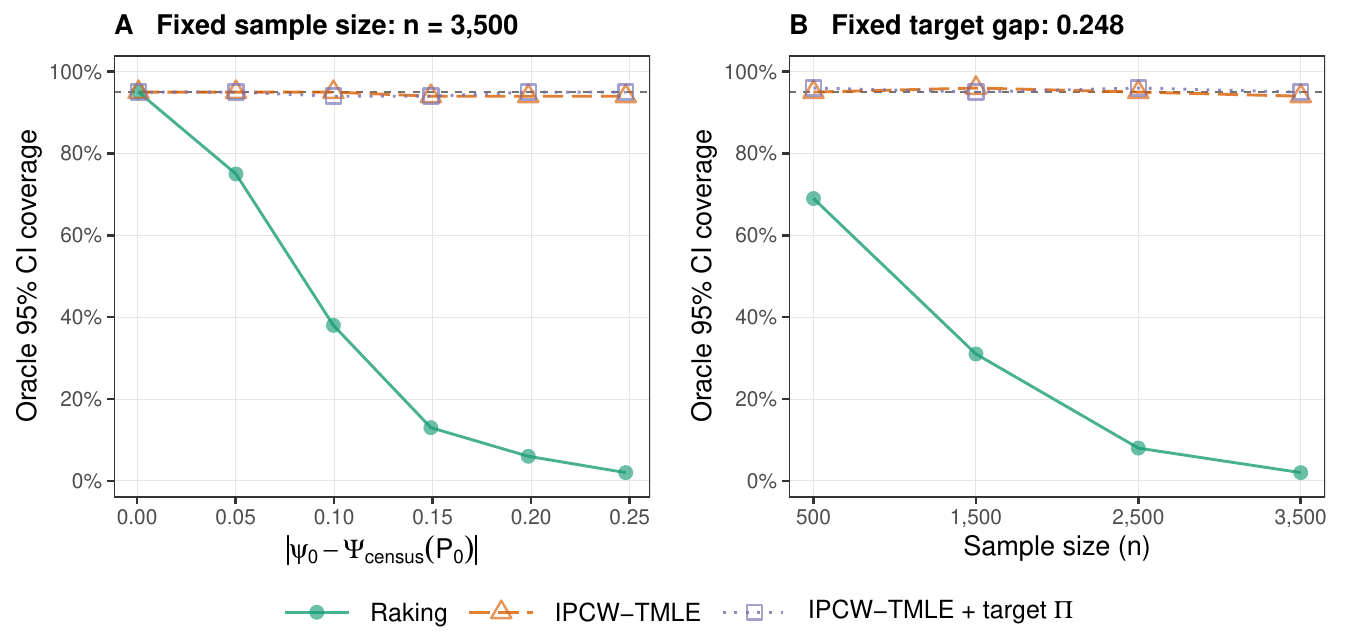}
\caption{Wald-type 95\% confidence interval oracle coverage of raking, IPCW-TMLE, and IPCW-TMLE with targeting of $\Pi$. Oracle coverage is the proportion of 1,000 Monte Carlo runs in which the interval computed using the empirical standard error covers $\psi_0=\theta_0$ in this simulation. Panel A fixes $n=3500$ and varies $|\psi_0-\Psi_{\mathrm{census}}(P_0)|$. Panel B fixes this gap at approximately $0.248$ and varies $n$. The census estimand is defined in~\eqref{eq:census_target}. Distinct marker shapes, partial opacity, and line types distinguish overlapping curves; points retain their original coordinates. Table~\ref{tab:raking_census_estimand} instead evaluates coverage of the census estimand.}
\label{fig:raking_cover}
\end{figure}

\begin{samepage}
As the gap between the two targets increases, the raking interval in Figure~\ref{fig:raking_cover} covers $\psi_0$ less often. At a fixed gap of $0.248$, Table~\ref{tab:raking_census_estimand} shows small bias relative to the census estimand and near-nominal coverage of that target.
\end{samepage}
\begin{table}[ht]
\centering
\resizebox{16cm}{!}{
\begin{tabular}{ccccccc}
\toprule
Estimator & $n$ & $|\text{Bias}|$ ($\times 10^{-3}$) & SE ($\times 10^{-2}$) & MSE ($\times 10^{-3}$) & Coverage (\%) & Oracle Coverage (\%) \\
\midrule
\multirow{4}{*}{Raking}
 & 500  & 2.10 & 16.38 & 26.78 & 95 & 96 \\
 & 1500 & 10.41 & 9.53 & 9.17 & 95 & 95 \\
 & 2500 & 3.82 & 7.23 & 5.23 & 95 & 95 \\
 & 3500 & 5.21 & 6.23 & 3.91 & 95 & 95 \\
\bottomrule
\end{tabular}
}
\caption{Performance metrics of the raking estimator with respect to the census estimand when $|\psi_0-\Psi_{\mathrm{census}}(P_0)|=0.248$. ``SE'' is the empirical standard error. ``MSE'' is the mean squared error. ``Coverage'' is the proportion of runs in which the 95\% confidence interval computed using the analytic variance estimator covers the census estimand. ``Oracle coverage'' is the proportion of runs in which the 95\% confidence interval computed using the empirical standard error covers the census estimand.}
\label{tab:raking_census_estimand}
\end{table}
\end{remark}

Raking calibration itself does not require a parametric final outcome regression. To use calibration within TMLE, choose the centered EIC regression $m_n(V)=E_n(D_n^F\mid\Delta=1,V)$ from~\eqref{eq:conditional_eic_regressions} as the auxiliary variable. The corresponding constraint is
$$
\frac{1}{n}\sum_{i=1}^n\left\{\frac{\Delta_i a_i}{\Pi_n(V_i)}-1\right\}m_n(V_i)=0.
$$
This solves the empirical equation for the sampling-mechanism component of the EIC at the current nuisance fits. Calibration using $\widehat h_{A,i}$ alone does not generally solve this equation. Our IPCW-TMLE with raking of $\Pi$ uses $m_n$ and retains the outcome-targeting step described in Section~\ref{sec:ipcw_tmle_target_Pi_local}. It is a separate estimator from the parametric regression-based raking estimator above. We observe no notable differences between raking and submodel targeting of $\Pi$ in our simulations. A theoretical comparison lies beyond the scope of this article.

Possible extensions should distinguish changes to the auxiliaries from changes to the final target. More flexible imputation or regression methods could improve the calibration auxiliaries and efficiency. They do not, by themselves, remove a gap between the census estimand and $\psi_0$ when the final outcome regression is unchanged. For estimation of $\psi_0$, one could use the nonparametric EIC within TMLE and a flexible outcome learner such as the highly adaptive lasso \citep{vdl_generally_2017,benkeser_hal_2016}. Alternatively, an adaptive TMLE could target an explicitly defined data-adaptive projection parameter \citep{vdl_adaptive_2023}. Such a change requires stating the new target and its relation to $\psi_0$. We leave these extensions for future work.

\section{Two estimators that rely on a slight rearrangement of the A-IPCW representation of the EIC}\label{sec:method_aipcw_rearrange}
The IPCW-TMLE variants reviewed above use the A-IPCW representation of the EIC\@. We now express this EIC in an algebraically equivalent form that motivates two alternative estimator constructions. The efficient estimating-equation (EEE) estimator targets the conditional expectation of the uncentered full-data EIC and averages the targeted predictions to estimate the target parameter. The Quasi-TMLE additionally imposes a plug-in constraint. Both constructions use the same observed-data EIC and efficiency bound as the preceding methods.

Recall the A-IPCW representation of the EIC:
$$
D_P=\frac{\Delta D^F_P}{\Pi_P(V)}-\frac{E_P(D^F_P\mid\Delta=1,V)}{\Pi_P(V)}(\Delta-\Pi_P(V)),
$$
which, with slight rearrangement, can also be written as:
$$
D_P=\frac{\Delta}{\Pi_P(V)}\left(\bar{D}^F_P-E_P(\bar{D}^F_P\mid\Delta=1,V)\right)+E_P(\bar{D}^F_P\mid\Delta=1,V)-\Psi(P),
$$
Here $\bar{D}^F_P$ and $\bar m_P$ are the uncentered EIC and its conditional regression defined in Section~\ref{sec:eic}. Thus the same representation is $D_P=\Delta\{\bar D_P^F-\bar m_P(V)\}/\Pi_P(V)+\bar m_P(V)-\Psi(P)$. This representation yields the following constructions.

\subsection{Efficient estimating-equation (EEE) estimator}\label{sec:eee}
Let $\bar m_n$ denote the fitted regression of $\bar D_n^F$ on $V$ among phase-2 participants, evaluated for everyone. Suppose that it is targeted to give $\bar m_n^\star$ satisfying
$$
0=P_n\left\{\frac{\Delta}{\Pi_n(V)}\left(\bar{D}^F_n-\bar m_n^\star\right)\right\}.
$$
Define $\hat\psi^{\mathrm{EEE}}\coloneq P_n\bar m_n^\star$. The fitted EIC is $D_n=\Delta(\bar D_n^F-\bar m_n^\star)/\Pi_n+\bar m_n^\star-\hat\psi^{\mathrm{EEE}}$, so its empirical mean is zero. The estimator averages the targeted conditional predictions over the full cohort. Although it targets a regression, this construction does not impose equality with the standardized contrast $\Psi(Q_{n,\epsilon},P_{W,n})$. The Quasi-TMLE below adds that plug-in constraint.

\subsection{Quasi-TMLE}\label{sec:quasi_tmle}
Consider a logistic fluctuation submodel for $Q_n$ given by
$$
\text{logit }Q_{n,\epsilon}(A,W)=\text{logit }Q_n(A,W)+\epsilon H_n(A,W).
$$
This logistic fluctuation submodel implies a submodel for $\bar{D}^F_n$ given by $\bar{D}_{n,\epsilon}^F\equiv\bar{D}_n^F(Q_{n,\epsilon})$. Using the notation in~\eqref{eq:conditional_eic_regressions}, let $\bar m_n$ be the initial uncentered EIC regression and $\bar{m}_{n,\epsilon}=E_n(\bar{D}^F_{n,\epsilon}\mid\Delta=1,V)$ its counterpart at $Q_{n,\epsilon}$. We define a one-dimensional submodel for $\bar{m}_{n,\epsilon}$ given by $\bar{m}_{n,\epsilon,\gamma}=\bar{m}_{n,\epsilon}+\gamma\Delta/\Pi_n$ with parameter $\gamma$. Recall that our goal is to solve the empirical equation for the following EIC\@:
$$
D_P=\frac{\Delta}{\Pi_P(V)}\{\bar D_P^F-\bar m_P(V)\}+\bar m_P(V)-\Psi(P).
$$
Define our estimator as $\Psi(P_n^\star)=P_n\bar{m}_{n,\epsilon,\gamma}$. The system of equations we need to solve is then
\begin{align*}
\left\{
\begin{array}{ll}
0=P_n\{\Delta/\Pi_n(\bar{D}_{n,\epsilon}^F-\bar{m}_{n,\epsilon,\gamma})\}\\
0=P_n\bar{m}_{n,\epsilon,\gamma}-\Psi(Q_{n,\epsilon},P_{W,n})\quad\text{for all }\epsilon.
\end{array}
\right.
\end{align*}
where the first equation, together with the definition of our estimator, solves the empirical mean of the EIC. The second equation ensures that our estimator is a plug-in estimator.
\begin{remark}
To solve for $\epsilon$ and $\gamma$, one can first express $\gamma$ as a function of $\epsilon$, denoted $\gamma(\epsilon)$. Using the second equation, $\gamma(\epsilon)$ can be derived and substituted into the first equation to solve for $\epsilon$. This reduces the problem to a one-dimensional root-finding task, which can be solved using, for example, the secant method (\cite{papakonstantinou_origin_2013}). Interestingly, a similar TMLE that uses numerical methods to solve the empirical mean of the EIC was proposed in \cite{chaffee_targeted_2011}, which in their longitudinal data settings is a computationally appealing alternative to the standard TMLE.
\end{remark}

\section{A TMLE using an alternative representation of the target parameter}\label{sec:real_tmle}
We next express the target through conditional expectations and use this representation to specify separate targeting steps. This construction uses the same TMLE framework and observed-data EIC as the preceding methods. For $P\in\mathcal{M}$, write $P_X=P_X(P)$ for the recovered full-data distribution in~\eqref{eq:recover_full_data}. Iterated expectation gives
$$
\Psi(P)=\Psi^F(P_X)=E_{P_V}E_{\Gamma_P}[Q_P(1,W)-Q_P(0,W)\mid V].
$$
The inner expectation integrates over the full-data conditional distribution given $V$, identified from phase-2 participants under CAR\@. This is the same statistical functional as~\eqref{eq:observed_data_target}, expressed through $P_V$ and $\Gamma_P$. It does not require an additional causal assumption. Lemma~\ref{lem:can_grad_of_psi_new} gives its canonical gradient.
\begin{lemma}\label{lem:can_grad_of_psi_new}
For $P\in\mathcal{M}$, let $P_X=P_X(P)$ and $\Gamma_P=P_{X\mid V}$ as in Section~\ref{sec:notation}. Define 
$$
H_{g,P_X}(A,W)=\frac{A}{g_{P_X}(W)}-\frac{1-A}{1-g_{P_X}(W)}
$$
The canonical gradient of the parameter $\Psi:\mathcal{M}\rightarrow\RR$ at $P$, where
$$
\Psi(P)=E_{P_V}E_{\Gamma_P}[Q_P(1,W)-Q_P(0,W)\mid V]
$$
is given by
$$
D^\star_{\Psi,P}=D_{Q,P}+D_{\Pi,P}+D_{\Gamma,P}+D_{P_V,P},
$$
where
\begin{align*}
D_{Q,P}&=\frac{\Delta}{\Pi_P(V)}\left[H_{g,P_X}(A,W)(Y-Q_{P_X}(A,W))\right]\\
D_{\Pi,P}&=-\frac{\Delta-\Pi_P(V)}{\Pi_P(V)}E\left[H_{g,P_X}(A,W)(Y-Q_{P_X}(A,W))\mid\Delta=1,V\right],\\
D_{\Gamma,P}&=\frac{\Delta}{\Pi_P(V)}\left[Q_{P_X}(1,W)-Q_{P_X}(0,W)-E_{P_{X\mid V}}[Q_{P_X}(1,W)-Q_{P_X}(0,W)\mid\Delta=1,V]\right],\\
D_{P_V,P}&=E_{P_{X\mid V}}[Q_{P_X}(1,W)-Q_{P_X}(0,W)\mid\Delta=1,V]-\Psi(P).
\end{align*}
\end{lemma}
The derivation is provided in Appendix C. 
\begin{remark}
The canonical gradient in Lemma~\ref{lem:can_grad_of_psi_new} is the established A-IPCW EIC from Section~\ref{sec:eic}, written in components that determine the updates below. The full-data to observed-data mapping is due to \citet{robins_estimation_1994} and Theorem~7.2 of \citet{tsiatis_semiparametric_2006}; see also \citet{kennedy_missing_2020}. We retain the derivation to explain this particular decomposition and targeting procedure.
\end{remark}

The components in Lemma~\ref{lem:can_grad_of_psi_new} suggest the targeting steps. Algorithm~\ref{alg:real_tmle} first targets $Q$ to solve the empirical mean of $D_{Q,P}$ component, then targets $\Pi$ to solve the empirical mean of $D_{\Pi,P}$ component, and alternates these updates with $g_n$ fixed. Its inner loop checks $\left|P_n\{D_{Q,P_n^{(k)}}+D_{\Pi,P_n^{(k)}}\}\right|\leq s_n^{(k)}$. The tolerance has the form in~\eqref{eq:adaptive_tolerance}, but here $\sigma_n^{(k)}$ is the empirical standard deviation of the sum $D_{Q,P_n^{(k)}}+D_{\Pi,P_n^{(k)}}$ over everyone. After the inner loop, regress $Q_n^{(K)}(1,W)$ and $Q_n^{(K)}(0,W)$ on $V$ among phase-2 participants and predict for everyone. Algorithm~\ref{alg:real_tmle} writes these fitted predictions directly as $E_{\Gamma_n}[Q_n^{(K)}(a,W)\mid\Delta=1,V]$ and their updated values as $E_{\Gamma_n}^\star[Q_n^{(K)}(a,W)\mid\Delta=1,V]$, for $a=0,1$. The two conditional-mean fluctuations solves the empirical mean of $D_{\Gamma,P}$ component. The factor $\Delta$ leaves prediction entries for phase-1-only participants unchanged, while the final average includes everyone and gives a zero empirical mean for $D_{P_V,P}$. Appendix~\ref{sec:targeting_implementation} describes implementation details of the alternating $Q,\Pi$ updates.

\begin{algorithm}
\caption{TMLE estimator}
\begin{algorithmic}[1]\label{alg:real_tmle}
\STATE{} Obtain initial estimator $\Pi_n^{(0)}$ for $\Pi_0$ using everyone;
\STATE{} Obtain initial estimators $Q_n^{(0)}$ and $g_n$ using phase-2 data with weights $\Delta/\Pi_n^{(0)}$;
\STATE{} Compute the initial sum $D_{Q,P_n^{(0)}}+D_{\Pi,P_n^{(0)}}$ and its empirical standard deviation $\sigma_n^{(0)}$ over everyone. Set $s_n^{(0)}=\sigma_n^{(0)}/(\sqrt n\log n)$;
\STATE{} Set iteration counter $k=0$;
\WHILE{$|P_n D_{Q,P_n^{(k)}}+P_n D_{\Pi,P_n^{(k)}}|>s_n^{(k)}$}
\STATE{} To target $D_Q$, fit a logistic regression without an intercept on phase-2 data, with outcome $Y$, offset $\operatorname{logit}Q_n^{(k)}(A,W)$, covariate $H_n(A,W)$, and weights $\Delta/\Pi_n^{(k)}(V)$;
\STATE{} With fitted coefficient $\epsilon_n^{(k)}$, update $Q_n^{(k)}$:
$$
\operatorname{logit}Q_n^{(k+1)}(A,W)=\operatorname{logit}Q_n^{(k)}(A,W)+\epsilon_n^{(k)}H_n(A,W);
$$
\STATE{} To target $D_\Pi$, fit a univariate logistic regression using everyone, with outcome $\Delta$, offset $\operatorname{logit}\Pi_n^{(k)}(V)$, and covariate
$$
\tilde C_n^{(k)}(V)=\frac{E_{\Gamma_n}[H_n(A,W)\{Y-Q_n^{(k+1)}(A,W)\}\mid\Delta=1,V]}{\Pi_n^{(k)}(V)},
$$
fitting the conditional regression on phase-2 data and predicting for everyone;
\STATE{} With fitted coefficient $\delta_n^{(k)}$, update $\Pi_n^{(k)}$:
$$
\operatorname{logit}\Pi_n^{(k+1)}(V)=\operatorname{logit}\Pi_n^{(k)}(V)+\delta_n^{(k)}\tilde C_n^{(k)}(V);
$$
\STATE{} Recompute $D_{Q,P_n^{(k+1)}}+D_{\Pi,P_n^{(k+1)}}$ and its empirical standard deviation $\sigma_n^{(k+1)}$ over everyone. Set $s_n^{(k+1)}=\sigma_n^{(k+1)}/(\sqrt n\log n)$ as in~\eqref{eq:adaptive_tolerance};
\STATE{} Increase $k$ by one;
\ENDWHILE{}
\STATE{} Set $K=k$. To target $D_\Gamma$, fit a univariate logistic regression on phase-2 data with outcome $Q_n^{(K)}(1,W)$, offset $\operatorname{logit}E_{\Gamma_n}[Q_n^{(K)}(1,W)\mid\Delta=1,V]$, and covariate $1/\Pi_n^{(K)}(V)$;
\STATE{} With fitted coefficient $\gamma_{1,n}$, update
$$
\operatorname{logit}E_{\Gamma_n}^\star[Q_n^{(K)}(1,W)\mid\Delta=1,V]=\operatorname{logit}E_{\Gamma_n}[Q_n^{(K)}(1,W)\mid\Delta=1,V]+\gamma_{1,n}\frac{\Delta}{\Pi_n^{(K)}(V)};
$$
\STATE{} For treatment level $0$, fit a univariate logistic regression on phase-2 data with outcome $Q_n^{(K)}(0,W)$, offset $\operatorname{logit}E_{\Gamma_n}[Q_n^{(K)}(0,W)\mid\Delta=1,V]$, and covariate $1/\Pi_n^{(K)}(V)$;
\STATE{} With fitted coefficient $\gamma_{0,n}$, update
$$
\operatorname{logit}E_{\Gamma_n}^\star[Q_n^{(K)}(0,W)\mid\Delta=1,V]=\operatorname{logit}E_{\Gamma_n}[Q_n^{(K)}(0,W)\mid\Delta=1,V]+\gamma_{0,n}\frac{\Delta}{\Pi_n^{(K)}(V)};
$$
\STATE{} Compute the final estimate by averaging over everyone, which also solves the $D_{P_V}$ equation:
$$
\Psi(P_n^\star)=P_n\left\{E_{\Gamma_n}^\star[Q_n^{(K)}(1,W)\mid\Delta=1,V]-E_{\Gamma_n}^\star[Q_n^{(K)}(0,W)\mid\Delta=1,V]\right\}.
$$
\end{algorithmic}
\end{algorithm}

\section{Analyzing the exact remainder terms}\label{sec:analyze_exact_rem}
We study the exact remainder for the statistical target $\Psi$. For each $P\in\mathcal{M}$, set $P_X=P_X(P)$ and use $Q_P$, $g_P$, and $D_P^F$ as defined in Section~\ref{sec:notation}.

\subsection{Robustness structure of the exact remainder term}
Recall the exact remainder from~\eqref{eq:exact_remainder_definition},
$$
R(P,P_0)=\Psi(P)-\Psi(P_0)+P_0D_{P}.
$$
By definition and iterated expectations, we have
\begin{align*}
R(P,P_0)&=P_0\left\{\frac{\Delta}{\Pi_P(V)}\left[\bar{D}^F_P-E_P(\bar{D}^F_P\mid\Delta=1,V)\right]+E_P(\bar{D}^F_P\mid\Delta=1,V)\right\}-\Psi(P_0)\\
&=P_0\left\{\frac{\Pi_0}{\Pi_P}(V)\left[E_{P_0}(\bar{D}^F_{P}\mid\Delta=1,V)-E_{P}(\bar{D}^{F}_{P}\mid\Delta=1,V)\right]+E_{P}(\bar{D}^F_{P}\mid\Delta=1,V)\right\}-\Psi(P_0)\\
&=P_0\left\{\frac{\Pi_P-\Pi_0}{\Pi_P}(V)\left[E_{P}(\bar{D}^F_{P}\mid\Delta=1,V)-E_{P_0}(\bar{D}^{F}_{P}\mid\Delta=1,V)\right]+E_{P_0}(\bar{D}^F_{P}\mid\Delta=1,V)\right\}-\Psi(P_0).
\end{align*}
The full-data remainder is defined under $P_{X,0}$:
$$
R^F(P_X,P_{X,0})\coloneq\Psi^F(P_X)-\Psi^F(P_{X,0})+E_{P_{X,0}}\{D^F_{P_X}(X)\}.
$$
Since $\bar{D}^F_P=D^F_{P_X}+\Psi^F(P_X)$, CAR and iterated expectation imply
$$
\begin{aligned}
R^F(P_X,P_{X,0})&=E_{P_{X,0}}\{\bar{D}^F_P(X)\}-\Psi^F(P_{X,0})\\
&=E_{P_{0,V}}\{E_{P_0}(\bar{D}^F_P\mid\Delta=1,V)\}-\Psi(P_0).
\end{aligned}
$$
Therefore, we have that
$$
R(P,P_0)=P_0\left\{\frac{\Pi_P-\Pi_0}{\Pi_P}(V)\left[E_{P}(\bar{D}^F_{P}\mid\Delta=1,V)-E_{P_0}(\bar{D}^{F}_{P}\mid\Delta=1,V)\right]\right\}+R^F(P_X,P_{X,0}),
$$
where
$$
R^F(P_X,P_{X,0})=E_{P_{X,0}}\left\{\{g_P(W)-g_0(W)\}\left[\frac{(Q_P-Q_0)(1,W)}{g_P(W)}+\frac{(Q_P-Q_0)(0,W)}{1-g_P(W)}\right]\right\}.
$$
This is the general decomposition used in the proof of Theorem~1 of \citet{hejazi_efficient_2021} (Supplementary Section~S4). Recall $m_P$ from~\eqref{eq:conditional_eic_regressions}. To make the correspondence explicit, define
$$
m_{0,P}(V)\coloneq E_{P_0}(D_P^F\mid\Delta=1,V),\qquad \bar m_{0,P}(V)\coloneq E_{P_0}(\bar D_P^F\mid\Delta=1,V).
$$
Then the identity above becomes
\begin{equation}\label{eq:remainder_literature_comparison}
R(P,P_0)=R^F(P_X,P_{X,0})+E_{P_0}\left[\frac{\Pi_P(V)-\Pi_0(V)}{\Pi_P(V)}\{m_P(V)-m_{0,P}(V)\}\right].
\end{equation}
At $P=\widehat P_n$, their $C$, $g_C$, $G_n$, and $\widetilde{G}_0$ correspond to our $\Delta$, $\Pi$, $m_{\widehat P_n}$, and $m_{0,\widehat P_n}$, respectively. Both conditional expectations hold the candidate full-data EIC fixed; in particular, $m_{0,\widehat P_n}$ is not the conditional expectation of the EIC evaluated at $P_0$. Since $\bar m_{0,P}=m_{0,P}+\Psi(P)$, the same error is $\bar m_P-\bar m_{0,P}=m_P-m_{0,P}$. This is the uncentered regression error used below. We retain the derivation for completeness and specialize it to our standardized contrast. Lemma~\ref{lem:exact_rem} gives an explicit bound, distinguishing observed-data and full-data norms and allowing the positivity lower bounds to depend on $n$.

The remainder has two error products: the sampling-mechanism error times the error in the conditional expectation of $\bar{D}^F_P$, and the treatment-mechanism error times the outcome-regression error. This structure gives the remainder conditions for IPCW-TMLE with targeting of $\Pi$, the EEE estimator, and Quasi-TMLE\@. Under the other regularity conditions, asymptotic efficiency follows when both products, with their inverse-probability factors, are $o_P(n^{-1/2})$. One nuisance function may converge more slowly if its paired function converges sufficiently fast.
\begin{lemma}\label{lem:exact_rem}
For the exact remainder
$$
R(P,P_0)=P_0\left\{\frac{\Pi_P-\Pi_0}{\Pi_P}(V)\left[E_{P}(\bar{D}^F_{P}\mid\Delta=1,V)-E_{0}(\bar{D}^{F}_{P}\mid\Delta=1,V)\right]\right\}+R^F(P_X,P_{X,0}),
$$
where
$$
R^F(P_X,P_{X,0})=E_{P_{X,0}}\left\{\{g_P(W)-g_0(W)\}\left[\frac{(Q_P-Q_0)(1,W)}{g_P(W)}+\frac{(Q_P-Q_0)(0,W)}{1-g_P(W)}\right]\right\},
$$
define the conditional-expectation error
$$
e_n(V)\coloneq\bar m_{\widehat P_n}(V)-\bar m_{0,\widehat P_n}(V).
$$
Use the observed-data norm $\norm{f}_{P_0}^2\coloneq E_{P_0}\{f^2(O)\}$. For functions of baseline covariates and functions of treatment and baseline covariates, respectively, define
$$
\norm{f}_{W,0}^2\coloneq E_{P_{X,0}}\{f^2(W)\},\qquad \norm{h}_{A,W,0}^2\coloneq\sum_{a=0}^1E_{P_{X,0}}\{h^2(a,W)\}.
$$
If $\mu_{g,n}\leq g_n(W)\leq1-\mu_{g,n}$ for $\mu_{g,n}>0$ and $\Pi_n(V)\geq\mu_{\Pi,n}>0$ almost surely, with probability tending to one, Cauchy-Schwarz gives
$$
|R(\widehat P_n,P_0)|\leq\mu_{\Pi,n}^{-1}\norm{\Pi_n-\Pi_0}_{P_0}\norm{e_n}_{P_0}+\sqrt{2}\,\mu_{g,n}^{-1}\norm{g_n-g_0}_{W,0}\norm{Q_n-Q_0}_{A,W,0}.
$$
The factor $\sqrt{2}$ accounts for the two outcome-regression errors paired with the same scalar propensity-score error. For example, suppose that, for $\alpha_1,\alpha_2,\alpha_3,\alpha_4>0$,
$$
\begin{aligned}
\norm{\Pi_n-\Pi_0}_{P_0}&=O_P(n^{-1/4-\alpha_1}),&\norm{e_n}_{P_0}&=O_P(n^{-1/4-\alpha_2}),\\
\norm{g_n-g_0}_{W,0}&=O_P(n^{-1/4-\alpha_3}),&\norm{Q_n-Q_0}_{A,W,0}&=O_P(n^{-1/4-\alpha_4}).
\end{aligned}
$$
Then
$$
R(\widehat P_n,P_0)=O_P\left(\mu_{\Pi,n}^{-1}n^{-1/2-\alpha_1-\alpha_2}+\mu_{g,n}^{-1}n^{-1/2-\alpha_3-\alpha_4}\right).
$$
For $\mu_{g,n}=n^{-\beta_1}$ and $\mu_{\Pi,n}=n^{-\beta_2}$, it suffices that $0<\beta_1<\alpha_3+\alpha_4$ and $0<\beta_2<\alpha_1+\alpha_2$ to obtain $R(\widehat P_n,P_0)=o_P(n^{-1/2})$.
\end{lemma}
The paired-error structure also underlies the multiple-robustness result in Lemma~1 of \citet{hejazi_efficient_2021}. Here Lemma~\ref{lem:exact_rem} gives four sufficient scenarios for $R(\widehat P_n,P_0)=0$: (i) $\Pi_n=\Pi_0$ and $g_n=g_0$; (ii) $\Pi_n=\Pi_0$ and $Q_n=Q_0$; (iii) $e_n=0$ and $g_n=g_0$; (iv) $e_n=0$ and $Q_n=Q_0$. Using $e_n$ makes explicit that the conditional-expectation error concerns the same candidate EIC on both sides. With estimated nuisance functions, consistency requires these products to vanish together with the other regularity conditions; efficient inference requires the stronger rate conditions discussed above. These remainder conditions concern the final estimator and do not automatically ensure that the initial learning procedures deliver the required nuisance limits. Ordinary IPCW learning generally relies on consistently estimated sampling probabilities to recover the full-data risk. This is a sufficient route to valid nuisance learning, not a necessary condition for every regression or sampling distribution. Section~\ref{sec:dr_initial_learning} describes how to target the sampling weights when learning $Q$ and $g$.

For the TMLE in Section~\ref{sec:real_tmle}, Lemma~\ref{lem:R2_new} gives a componentwise refinement of~\eqref{eq:remainder_literature_comparison}. It follows by splitting $\bar{D}_P^F$ into the outcome-residual term $D^F_{Q,P}$ and the contrast $Q_P(1,W)-Q_P(0,W)$. This separates the conditional-expectation errors corresponding to the updates in that section. All full-data quantities are evaluated under $P_X(P)$, as above.
\begin{lemma}\label{lem:R2_new}
Let 
$$
D^F_{Q,P}=\left[\frac{A}{g_P(W)}-\frac{1-A}{1-g_P(W)}\right](Y-Q_P(A,W))\quad\text{and}\quad
\mu_{P_1,P_2}^{(a)}=E_{P_1}[Q_{P_2}(a,W)\mid\Delta=1,V],
$$
for $a\in\{0,1\}$. The exact remainder of $\Psi(P)=E_{P_V}E_{\Gamma_P}[Q_P(1,W)-Q_P(0,W)\mid V]$ is given by
\begin{align*}
R(P,P_0)&=P_0\left\{\frac{\Pi_P-\Pi_0}{\Pi_P}(V)\left[E_{P}(D^F_{Q,P}\mid\Delta=1,V)-E_{0}(D_{Q,P}^F\mid\Delta=1,V)\right]\right\}\\
&+E_{P_{X,0}}\left\{\{g_P(W)-g_0(W)\}\left[\frac{(Q_P-Q_0)(1,W)}{g_P(W)}+\frac{(Q_P-Q_0)(0,W)}{1-g_P(W)}\right]\right\}\\
&+P_0\left\{\frac{\Pi_P-\Pi_0}{\Pi_P}(V)\left(\mu_{P,P}^{(1)}-\mu_{P_0,P}^{(1)}\right)-\frac{\Pi_P-\Pi_0}{\Pi_P}(V)\left(\mu_{P,P}^{(0)}-\mu_{P_0,P}^{(0)}\right)\right\}.
\end{align*}
\end{lemma}
The derivation of Lemma~\ref{lem:R2_new} is provided in Appendix C. In particular, note that in the exact remainder we have second-order differences in $(P-P_0)$, i.e., $(\Pi_P-\Pi_0)[E_P(D^F_{Q,P}\mid\Delta=1,V)-E_0(D^F_{Q,P}\mid\Delta=1,V)]$, $(g_P-g_0)(Q_P-Q_0)$, and $(\Pi_P-\Pi_0)(\mu_{P,P}^{(a)}-\mu_{P_0,P}^{(a)})$.

\subsection{Double-robust IPCW-loss based machine learning of full-data $Q$ and $g$}\label{sec:dr_initial_learning}
Recall that across our TMLE variants discussed in previous sections, the initial estimators of the outcome regression $Q$ and treatment mechanism $g$, both of which are full-data parameters, are fit using an IPCW-loss. Specifically, we first obtain an estimator $\Pi_n$ for the phase-2 sampling mechanism $\Pi_0$, and then weight observations by $\Delta/\Pi_n(V)$ when estimating $Q_0$ and $g_0$. This IPCW approach is attractive because it allows the use of any machine learning algorithm that accepts weights. However, a drawback is that the resulting estimators of $Q_0$ and $g_0$ depend on $\Pi_0$ being consistently estimated; thus, one cannot generally achieve true double robustness, since the consistency of the TMLE requires $\Pi$ to be consistently estimated. To overcome this limitation and attain true double-robust estimation even for the initial $Q_0$ and $g_0$, we propose the following iterative strategy.

For a candidate outcome regression $q$, let $R_0(q)$ be the true risk defined by $E_{P_{X,0}}\{L^F(q)(X)\}$ for a given full-data loss function $L^F(q)(X)$ such as the squared-error loss for continuous $Y$ or the log-likelihood loss for binary $Y$. We could estimate this full-data risk, which is just a mean of a full-data function, with an IPCW estimator using an estimator $\Pi_n$ for $\Pi_0$:
$$
R_n(q)=P_n\left[\frac{\Delta}{\Pi_n^\star(V)}L^F(q)(X)\right],
$$
where $\Pi_n^\star$ is a targeted estimator of $\Pi_n$ satisfying the score equation:
$$
0=P_n\left[\frac{m^L_{n,q}(V)}{\Pi_n^\star(V)}\{\Delta-\Pi_n^\star(V)\}\right],
$$
where $m^L_{n,q}(V)$ estimates $E_0\{L^F(q)(X)\mid\Delta=1,V\}$. Therefore, $R_n$ is also an A-IPCW estimator:
$$
R_n(q)=P_n\left[\frac{\Delta}{\Pi_n^\star(V)}L^F(q)(X)-\frac{m^L_{n,q}(V)}{\Pi_n^\star(V)}\{\Delta-\Pi_n^\star(V)\}\right].
$$
This motivates an iterative strategy which is applicable with any machine learning algorithm that accommodates weights. For illustration, consider the highly adaptive lasso (HAL), a flexible machine learning algorithm (\cite{vdl_generally_2017,benkeser_hal_2016}). Starting with an initial estimator $Q_n^{(0)}$ of $Q_0$, we target $\Pi_n$ toward $Q_n^{(0)}$ to obtain $\Pi_{n,Q_n^{(0)}}^\star$, and define the IPCW loss $R_n^{(k)}(q)=P_n[\{\Delta/\Pi_{n,Q_n^{(k)}}^\star(V)\}L^F(q)(X)]$.
We then update $Q_n^{(0)}$ by re-fitting an IPCW-HAL estimator $Q_n^{(1)}$ that minimizes $R_n^{(0)}(q)$. Iterating the targeting of $\Pi$ and re-fitting of $Q$ would yield a sequence $\{Q_n^{(k)}:k=0,\dots,K\}$ converging to a double-robust HAL estimator $Q_n$ of $Q_0$. In particular, if either $\Pi_0$ is consistently estimated or the conditional mean $E_0\{L^F(q)(X)\mid\Delta=1,V\}$ is consistently estimated, then our
HAL estimator will converge to the true target function at the typical HAL
rate as in the full-data model. The same iterative approach applies to the estimation of the treatment mechanism $g_0$ under the full-data log-likelihood loss $L_g^F(g)(X)$.
Although doubly robust risk estimation is useful when the sampling model may be misspecified, our simulations use ordinary IPCW initialization, and we leave evaluation of the iterative procedure to future work because incorporating it into the simulations requires additional implementation effort and computation.

\section{Simulations}\label{sec:simulation}
We compare the estimators in Table~\ref{tab:candidate_estimators} in two simulation settings. The first examines performance with supplied sampling and treatment probabilities and misspecified working regressions. We designed this setting to empirically assess the robustness properties of the estimators. The second estimates all nuisance functions, using logistic regressions for the sampling and treatment mechanisms and a discrete super learner for the outcome and auxiliary regressions. The reported bias and coverage use the intervention mean difference $\theta_0$ computed from each DGP\@. In both settings, treatment is conditionally exchangeable given $W$, and this contrast equals the full-data statistical target $\psi_0$.
\begin{table}[ht]
\centering
\resizebox{\textwidth}{!}{
\begin{tabular}{r|l}
\toprule
\textbf{Candidate estimator} & \textbf{Description} \\ \midrule
Raking & Raking with parametric outcome regression (Section~\ref{sec:raking})\\
A-IPCW & Augmented IPCW estimator\\
IPCW-TMLE & IPCW-TMLE without targeting of phase-2 sampling mechanism (Section \ref{sec:ipcw_tmle})\\
IPCW-TMLE target $\Pi$ & IPCW-TMLE with TMLE-based targeting of phase-2 sampling mechanism (Section \ref{sec:ipcw_tmle_target_Pi_local})\\
IPCW-TMLE rake $\Pi$ & IPCW-TMLE with raking-based targeting of phase-2 sampling mechanism (Section~\ref{sec:raking})\\
EEE & Efficient estimating-equation estimator (Section \ref{sec:eee})\\
Quasi-TMLE & Plug-in version of EEE (Section \ref{sec:quasi_tmle})\\
TMLE & TMLE that uses an alternative representation of the target parameter (Section \ref{sec:real_tmle})\\
\bottomrule
\end{tabular}
}
\caption{Candidate estimators and their descriptions.}
\label{tab:candidate_estimators}
\end{table}

\begingroup
\clubpenalties1 10000
In both settings, $V=(W_1,W_2,A,Y)$ is observed for everyone, whereas $W_3,W_4$ are observed only when $\Delta=1$. Treatment and outcome are binary. The first experiment varies sample size from 500 to 2,000 at approximately $50\%$ missingness, while the second fixes $n=1000$ and varies missingness across approximately $20\%$, $50\%$, and $70\%$. Bias, empirical standard error, mean squared error, and coverage are computed across 500 Monte Carlo repetitions per scenario. The ``oracle coverage'' calculation substitutes the empirical standard error across repetitions for each run's estimated standard error. Appendix~\ref{sec:simulation_dgps} describes the calculation of the population targets and gives the software details.
\par
\endgroup

\subsection{Supplied probabilities and misspecified working regressions}\label{sec:double_robust_demo}
In this design, motivated by \citet{kang_demystifying_2007}, we draw $Z_1,\ldots,Z_4$ independently from $N(0,1)$ and set
$$
W_1=\exp(Z_1/2),\quad W_2=Z_2^3,\quad W_3={(0.6+Z_1Z_4/25)}^3,\quad W_4={(20+Z_3+Z_4)}^2.
$$
Given $Z=(Z_1,Z_2,Z_3,Z_4)$, we draw treatment $A$, outcome $Y$, and phase-2 membership $\Delta$ in that order as Bernoulli variables with probabilities
$$
\begin{aligned}
P(A=1\mid Z) &= \operatorname{expit}(-0.2Z_1-0.6Z_2+0.9Z_4),\\
P(Y=1\mid A,Z) &= \operatorname{expit}(-1+0.6Z_1-0.4Z_2+0.2Z_3-0.5Z_4+1.2A),\\
P(\Delta=1\mid A,Y,Z) &= \operatorname{expit}(-0.1Z_1+0.1Z_2).
\end{aligned}
$$
Approximately $50\%$ of participants have missing phase-2 covariates. Sampling satisfies CAR because $Z_1=2\log W_1$ and $Z_2=\sqrt[3]{W_2}$ are known for everyone. Among phase-2 participants, $W_3$ also gives $Z_1Z_4=25(\sqrt[3]{W_3}-0.6)$. Since $Z_1\ne 0$ almost surely, the full covariate vector $W$ determines all three predictors of treatment. Thus the supplied treatment probability is $g_0(W)$, and adjustment for $W$ identifies the intervention risk difference, $\psi_0=\theta_0\approx 0.244435$.

This design illustrates the robustness case with known $\Pi_0$ and $g_0$ in Section~\ref{sec:analyze_exact_rem}. We supply both probabilities to the estimators that use them. Raking uses the sampling probabilities alone. The working outcome regression is misspecified. For example, it enters $W_2$ linearly, whereas the outcome equation depends on $\sqrt[3]{W_2}$. Table~\ref{tab:Pi_g_correct} reports results from 500 repetitions at each of $n=500,1000,1500,2000$.

\begin{table}[ht]
\centering
\resizebox{16cm}{!}{
\begin{tabular}{crccccc}
\toprule
$n$ & Estimator & $|\text{Bias}|$ ($\times 10^{-3}$) & SE ($\times 10^{-2}$) & MSE ($\times 10^{-3}$) & Coverage (\%) & Oracle Coverage (\%) \\
\midrule
\multirow{8}{*}{500}
& Raking                     & 42.5 & 4.52 & 3.840 & 82.6 & 82.2 \\
& A-IPCW                     & 3.58 & 6.77 & 4.588 & 94.8 & 95.6 \\
& IPCW-TMLE                  & 3.84 & 6.73 & 4.532 & 94.6 & 95.0 \\
& IPCW-TMLE + target $\Pi$   & 3.17 & 6.72 & 4.513 & 94.4 & 95.2 \\
& IPCW-TMLE + rake $\Pi$     & 3.01 & 6.72 & 4.520 & 94.4 & 95.2 \\
& EEE                        & 3.47 & 6.77 & 4.583 & 94.6 & 95.6 \\
& Quasi-TMLE                 & 2.82 & 6.73 & 4.523 & 94.2 & 95.2 \\
& TMLE                       & 4.25 & 6.70 & 4.493 & 94.6 & 95.4 \\
\midrule
\multirow{8}{*}{1000}
& Raking                     & 43.1 & 3.49 & 3.077 & 71.8 & 76.6 \\
& A-IPCW                     & 0.915 & 4.94 & 2.438 & 94.4 & 95.2 \\
& IPCW-TMLE                  & 1.22 & 4.91 & 2.407 & 94.2 & 95.2 \\
& IPCW-TMLE + target $\Pi$   & 1.04 & 4.91 & 2.408 & 94.4 & 95.6 \\
& IPCW-TMLE + rake $\Pi$     & 1.09 & 4.91 & 2.410 & 94.4 & 95.8 \\
& EEE                        & 1.00 & 4.94 & 2.438 & 94.8 & 95.4 \\
& Quasi-TMLE                 & 1.11 & 4.90 & 2.397 & 94.4 & 95.8 \\
& TMLE                       & 0.119 & 4.90 & 2.397 & 94.4 & 95.6 \\
\midrule
\multirow{8}{*}{1500}
& Raking                     & 41.8 & 2.52 & 2.386 & 65.6 & 62.4 \\
& A-IPCW                     & 0.804 & 3.81 & 1.453 & 95.0 & 95.0 \\
& IPCW-TMLE                  & 0.787 & 3.75 & 1.407 & 95.0 & 94.8 \\
& IPCW-TMLE + target $\Pi$   & 0.832 & 3.76 & 1.412 & 95.4 & 94.8 \\
& IPCW-TMLE + rake $\Pi$     & 0.893 & 3.76 & 1.413 & 95.2 & 94.8 \\
& EEE                        & 0.877 & 3.80 & 1.444 & 95.2 & 95.2 \\
& Quasi-TMLE                 & 1.11 & 3.76 & 1.412 & 95.4 & 95.0 \\
& TMLE                       & 0.00913 & 3.75 & 1.406 & 95.4 & 94.8 \\
\midrule
\multirow{8}{*}{2000}
& Raking                     & 43.3 & 2.25 & 2.380 & 53.6 & 53.0 \\
& A-IPCW                     & 0.832 & 3.31 & 1.095 & 94.8 & 94.6 \\
& IPCW-TMLE                  & 1.09 & 3.26 & 1.060 & 95.4 & 94.6 \\
& IPCW-TMLE + target $\Pi$   & 1.02 & 3.26 & 1.061 & 95.4 & 94.6 \\
& IPCW-TMLE + rake $\Pi$     & 0.951 & 3.26 & 1.061 & 95.2 & 94.6 \\
& EEE                        & 0.828 & 3.31 & 1.094 & 95.0 & 94.6 \\
& Quasi-TMLE                 & 0.931 & 3.26 & 1.062 & 95.0 & 94.6 \\
& TMLE                       & 1.83 & 3.25 & 1.059 & 95.4 & 94.4 \\
\bottomrule
\end{tabular}
}
\caption{Simulation results with known sampling and treatment probabilities and misspecified working regressions, based on 500 repetitions per sample size. SE is the empirical standard error and MSE is mean squared error. Coverage is the percentage of estimated 95\% confidence intervals containing the risk difference $\psi_0=\theta_0$. Oracle coverage uses the empirical SE\@.}
\label{tab:Pi_g_correct}
\end{table}

For the main-term logistic working regression, the census estimand is $\Psi_{\mathrm{census}}(P_0)\approx 0.201554$, about 4.29 percentage points below $\psi_0$. Its average estimate was within 0.11 percentage points of the census estimand at every sample size. As the intervals narrowed, their coverage of $\psi_0$ fell from $82.6\%$ at $n=500$ to $53.6\%$ at $n=2000$ (Table~\ref{tab:Pi_g_correct}). At $n=2000$, the same intervals covered the census estimand in $95.4\%$ of repetitions. Figure~\ref{fig:raking_cover} illustrates the same distinction in the Gaussian example.

The other seven estimators had absolute bias below 0.43 percentage points. Coverage ranged from $94.2\%$ to $95.4\%$ for their reported intervals and from $94.4\%$ to $95.8\%$ for their oracle intervals. The two variants that targeted $\Pi$ had empirical standard errors within $0.2\%$ of ordinary IPCW-TMLE at every sample size. The variance comparison in~\eqref{eq:target_pi_variance_gain} and the accuracy of the conditional EIC regression both bear on the potential gain from targeting. The non-plug-in estimators, A-IPCW and EEE, had higher empirical MSE than each of the five plug-in estimators at all four sample sizes. The MSE advantage of Quasi-TMLE over EEE was small in this experiment.

\FloatBarrier{}
\subsection{Simulation scenarios with all nuisance functions estimated}\label{sec:simulated_data}
We draw $W_1,\ldots,W_4$ independently from $N(1,1)$. Given $W$, we draw $A$, $Y$, and $\Delta$ in that order as Bernoulli variables with respective probabilities
$$
\begin{aligned}
g_0(W) &= \operatorname{expit}(-0.2W_1-0.6W_2+0.2W_4),\\
Q_0(A,W) &= \operatorname{expit}(0.1W_1^2-0.01W_2^3+0.2W_3-0.1W_4+0.6A+0.5AW_2^2),\\
\Pi_0(V) &= \operatorname{expit}(\alpha+0.2W_1+0.2Y).
\end{aligned}
$$
Setting $\alpha$ to $1.1$, $-0.3$, or $-1.1$ gives approximately $20\%$, $50\%$, or $70\%$ missingness, respectively. Sampling depends on the outcome but satisfies CAR because both $W_1$ and $Y$ are observed for everyone. The polynomial terms and treatment interaction make a main-term logistic outcome regression misspecified. Treatment is conditionally exchangeable given $W$. Numerical integration gives $\psi_0=\theta_0\approx 0.2593$ and a census estimand of approximately $0.2416$ for the main-term logistic working regression. Each missingness setting has $n=1000$ and 500 repetitions.

The initial sampling and treatment mechanisms are fit using main-term logistic regressions. The outcome regression and auxiliary conditional regressions use a discrete super learner with four candidates: a main-term GLM, XGBoost \citep{chen_xgboost_2015}, random forest \citep{random_brieman_2001,wright_ranger_2017}, and Bayesian additive regression trees (BART) \citep{bart_chipman_2010,dorie_dbarts_2025}. Five-fold cross-validation selects a single candidate, using binomial negative log-likelihood for the outcome regression and squared-error loss for the auxiliary regressions. The initial $Q$ and $g$ fits use inverse sampling weights among phase-2 participants. The sampling regression uses the full cohort, and auxiliary regressions are fit among phase-2 participants and predict for the full cohort. Raking retains its main-term logistic outcome regression and imputation-based calibration described in Section~\ref{sec:raking}. Appendix~\ref{sec:simulation_dgps} records the software and learner settings.

The discrete super learner selects the candidate with the smallest estimated validation risk. The oracle selector instead chooses, among the same fitted candidates, the one with the smallest true prediction risk under the chosen loss. Under the conditions for the cross-validation oracle inequality and its asymptotic optimality result, the discrete selector performs asymptotically as well as that oracle within the specified library \citep{vdl_sl_2003}.

\begin{table}[ht]
\centering
\resizebox{16cm}{!}{
\begin{tabular}{crcccccc}
\toprule
Missing (\%) & Estimator & $|\text{Bias}|$ ($\times 10^{-3}$) & SE ($\times 10^{-2}$) & MSE ($\times 10^{-3}$) & $|\text{Bias}|/\text{SE}$ & Coverage (\%) & Oracle Coverage (\%) \\
\midrule
\multirow{10}{*}{20}
& Raking                    & 18.3 & 2.88 & 1.165 & 0.635 & 92 & 91 \\
& A-IPCW                    & 0.946 & 2.88 & 0.827 & 0.033 & 95 & 96 \\
& IPCW-TMLE                 & 0.189 & 3.15 & 0.990 & 0.006 & 96 & 95 \\
& IPCW-TMLE + target $\Pi$  & 0.230 & 2.87 & 0.821 & 0.008 & 95 & 95 \\
& IPCW-TMLE + rake $\Pi$    & 0.193 & 2.87 & 0.820 & 0.007 & 95 & 95 \\
& EEE                       & 0.916 & 2.87 & 0.824 & 0.032 & 95 & 96 \\
& Quasi-TMLE         & 0.299 & 2.87 & 0.825 & 0.010 & 95 & 96 \\
& TMLE                      & 0.194 & 2.87 & 0.821 & 0.007 & 95 & 96 \\
\midrule
\multirow{10}{*}{50}
& Raking                    & 18.2 & 2.95 & 1.200 & 0.619 & 92 & 91 \\
& A-IPCW                    & 2.71 & 3.06 & 0.943 & 0.089 & 94 & 95 \\
& IPCW-TMLE                 & 0.692 & 4.06 & 1.644 & 0.017 & 96 & 94 \\
& IPCW-TMLE + target $\Pi$  & 0.886 & 3.02 & 0.911 & 0.029 & 94 & 95 \\
& IPCW-TMLE + rake $\Pi$    & 0.783 & 3.03 & 0.917 & 0.026 & 94 & 95 \\
& EEE                       & 2.65 & 3.04 & 0.931 & 0.087 & 94 & 95 \\
& Quasi-TMLE         & 0.827 & 3.04 & 0.920 & 0.027 & 94 & 95 \\
& TMLE                      & 0.859 & 3.02 & 0.910 & 0.028 & 94 & 95 \\
\midrule
\multirow{10}{*}{70}
& Raking                    & 18.0 & 2.99 & 1.218 & 0.600 & 93 & 92 \\
& A-IPCW                    & 4.67 & 3.40 & 1.175 & 0.137 & 93 & 95 \\
& IPCW-TMLE                 & 0.459 & 5.57 & 3.091 & 0.008 & 94 & 95 \\
& IPCW-TMLE + target $\Pi$  & 0.648 & 3.39 & 1.144 & 0.019 & 93 & 95 \\
& IPCW-TMLE + rake $\Pi$    & 0.677 & 3.38 & 1.142 & 0.020 & 93 & 94 \\
& EEE                       & 5.06 & 3.35 & 1.146 & 0.151 & 93 & 95 \\
& Quasi-TMLE         & 1.343 & 3.33 & 1.108 & 0.040 & 93 & 95 \\
& TMLE                      & 0.585 & 3.38 & 1.143 & 0.017 & 92 & 95 \\
\bottomrule
\end{tabular}
}
\caption{Simulation results over 500 Monte Carlo runs per missingness setting at $n=1000$. Sampling and treatment mechanisms use logistic regressions; outcome and auxiliary regressions use the discrete super learner described in Section~\ref{sec:simulated_data}. ``SE'' is the empirical standard error; ``MSE'' is mean squared error; ``Coverage'' is the proportion of runs where the 95\% confidence interval based on the estimated efficient influence curve variance covers the DGP intervention contrast $\theta_0$; ``Oracle coverage'' uses the empirical SE.}
\label{tab:simulated_data}
\end{table}

Table~\ref{tab:simulated_data} compares the estimators across the three missingness settings. The raking estimator has a large bias-to-standard-error ratio and reduced oracle coverage of $\psi_0$ at sample size $1,000$. These metrics concern $\psi_0$, rather than the census target for which the estimator can be consistent as described in Section~\ref{sec:raking}.

Both updates of $\Pi$ reduced the empirical variance of IPCW-TMLE in Table~\ref{tab:simulated_data}. Relative to IPCW-TMLE without a sampling update, logistic targeting reduced variance by $17.1\%$, $44.6\%$, and $63.0\%$ at $20\%$, $50\%$, and $70\%$ missingness, respectively. Raking-based targeting gave corresponding reductions of $17.2\%$, $44.3\%$, and $63.1\%$. Lemma~\ref{lem:target_pi_variance_gain} helps explain why the gain can grow with missingness here. The three settings share the same distribution of $(W,A,Y)$, so $m_0(V)$ is the same in each. Lowering the sampling intercept increases the weight $\{1-\Pi_0(V)\}/\Pi_0(V)$ on ${m_0(V)}^2$ in~\eqref{eq:target_pi_variance_gain}. Using the true nuisance functions gives benchmark variance reductions of $20.5\%$, $50.8\%$, and $69.1\%$ for the augmented score $D_0$ relative to $U_0$. Both fitted targeting methods show the same pattern, with smaller reductions at each missingness level. This supports a precision gain from targeting $\Pi$ in this design.

When the missingness rate increases to 50\%, the IPCW-TMLE shows a slightly higher MSE than the raking estimator. However, the oracle coverage of $\psi_0$ by the raking estimator is low, whereas the IPCW-TMLE, despite being less efficient, maintains good bias-to-standard-error control and nominal oracle coverage near 95\%. Consistent with our earlier findings, the plug-in property appears beneficial. Specifically, the A-IPCW shows elevated MSE in both scenarios, and the EEE estimator shows increased MSE when missingness reaches 50\%.

\subsection{Additional experiments}\label{sec:simulation_scope}
Appendix~\ref{sec:appendix_local_linearization} examines the computational approximation in Remark~\ref{rem:linearization}. In Table~\ref{tab:linearization}, replacing repeated conditional EIC regressions by local linearization gives similar MSE and coverage over $n=500,1000,1500,2000$. The measured time savings are modest: the average run time changes from $7.64$ to $7.16$ seconds at $n=500$, and from $18.74$ to $18.07$ seconds at $n=2000$. In the separate experiment with extreme estimated treatment probabilities (Table~\ref{tab:linearization_compare_positivity}), refitting the conditional regression reduced MSE for the intervention risk difference $\theta_0$. The derivative used in that approximation contains squared inverse-propensity terms, which can make its conditional regression unstable. Together, the main and supplementary experiments examine working-regression misspecification, estimated nuisance functions, missingness levels, and the behavior of a computational approximation.

\section{Concluding remarks}\label{sec:conclusion}
We reviewed estimators of full-data statistical parameters from two-phase samples, including the raking estimator, the IPCW-TMLE, and a variant that also targets the sampling mechanism. CAR and sampling positivity connect the observed-data target $\Psi(P_0)$ to $\Psi^F(P_{X,0})$. For the standardized mean difference, the additional causal assumptions identify this value with the ATE $\theta_0$. Under the sampling-weight and regression regularity conditions in Section~\ref{sec:raking}, the raking estimator studied here consistently estimates its census target, including when the final outcome regression is misspecified. That target can differ from $\psi_0$, even when causal identification holds. Correct outcome-regression specification is sufficient, but not necessary, for the targets to agree. For estimation of $\psi_0$, this estimator is less robust to outcome-regression misspecification than TMLE when the treatment and sampling mechanisms are consistently estimated and the remaining consistency conditions hold. This comparison does not imply that raking calibration itself is generally inconsistent or inefficient.

Our estimator constructions build on the established theory for two-phase and coarsened data, including sampling-mechanism targeting and the remainder decomposition studied by \citet{hejazi_efficient_2021}. The EEE estimator targets the conditional expectation of the uncentered full-data EIC, and the Quasi-TMLE additionally enforces a plug-in constraint. A further TMLE separately targets conditional expectations in a representation of the target parameter. The contribution is in these specific constructions, their computational treatment, and their comparison with existing estimators. The EIC mapping and general remainder decomposition are inherited from prior theory; our remainder results specialize that decomposition and separate the errors relevant to the targeting steps. Under the stated regularity and nuisance-convergence conditions, the efficient estimators share the same first-order asymptotic distribution. In our simulations, the plug-in estimators generally achieve smaller MSE than their non-plug-in analogues.

We conclude with several remarks on the limitations of this study and directions for future research. First, as discussed in Section~\ref{sec:raking}, there may be opportunities to improve the raking estimator by relaxing its parametric assumptions. For instance, instead of targeting a fixed census estimand defined by a prespecified parametric regression class, one could target a parameter defined on a data-adaptively learned regression class, in line with recent developments in the adaptive TMLE literature~\cite{vdl_adaptive_2023,vdl_adaptive_2024}. Second, in this work we estimated the full-data treatment and outcome mechanisms using standard regression methods based on IPCW-weighted loss functions. This approach may be suboptimal when the phase-2 sample proportion is small, as phase-1 observations are effectively discarded due to their zero weights. Alternative strategies that more efficiently leverage phase-1 data, such as imputation-based methods, could lead to improved estimation of $g$ and $Q$. It would be particularly interesting to explore super learner-based ensembles that include both inverse-weighted and imputation-based candidate learners, using cross-validation to select the optimal candidate for a given dataset and task. Third, although we have developed and studied multiple estimators, our simulation designs were manually constructed and therefore limited in scope. Future work could involve benchmarking these estimators across a broader set of realistic scenarios to better understand the conditions under which each performs best. Finally, even with extensive simulations, it may remain challenging to determine which estimator performs best in any given context, especially since many of them are TMLE variants and thus asymptotically equivalent. In this regard, it would be valuable to investigate ensemble strategies for combining estimators, as also discussed in~\cite{barnatchez_efficient_2025}, which enable data-adaptive selection among a class of asymptotically equivalent estimators.

\newpage
\bibliography{references}

\newpage
\appendix
\section{Implementation details}
\subsection{Generalized raking estimator}\label{sec:appendix_raking}
The implementation used in our simulations adapts the approach of \citet{williamson_assessing_2026}. Their code is available at \url{https://github.com/PamelaShaw/Missing-Confounders-Methods}, accessed July 28, 2025. Our simulations use $M=20$ MICE imputations. To construct calibration auxiliaries for everyone, the implementation also imputes copies of phase-2 observations with their phase-2 covariates masked. Thus the completed covariates used for this step are generated from the available phase-1 variables $V$ for both sampled and unsampled individuals.

In each imputed dataset, a main-term regression of $Y$ on $A,W_1,W_2,W_3,W_4$ produces influence contributions for the coefficient on $A$. The implementation averages these contributions across imputations to form $\widehat h_{A,i}$. A common scaling factor of $1/n$ in these contributions has no effect on the calibration constraint. The \texttt{survey::calibrate} function applies exponential raking to an intercept and $\widehat h_{A,i}$. The final fit uses \texttt{survey::svyglm} with the same main-term formula, an identity link for continuous outcomes, and a logit link for binary outcomes. The quasibinomial fit used for binary outcomes solves the same coefficient equations as weighted Bernoulli likelihood fitting. The \texttt{marginaleffects::avg\_comparisons} function then computes the difference between predictions at $A=1$ and $A=0$, averaged over phase-2 observations with normalized calibrated weights. This is~\eqref{eq:raking_estimate}. The coefficient influence function used for calibration need not equal the influence function of this standardized contrast.

\subsection{Inverse probability of censoring weighted targeted maximum likelihood estimation with phase-2 sampling mechanism targeted using raking}
For the IPCW-TMLE with raking of $\Pi$, the auxiliary variable is the fitted centered EIC regression from~\eqref{eq:conditional_eic_regressions}. During one sampling update, write $h_i=m_n(V_i)$ and hold these predictions fixed. Define $w_i^{(0)}=\Delta_i/\Pi_n^{(0)}(V_i)$ and $w_i^{(1)}=w_i^{(0)}a_i$ with $a_i>0$. Using the distance $d(p,q)=p\log(p/q)-p+q$ gives
\begin{equation}\label{eqn:raking_program}
\begin{aligned}
\min_{a_1,\ldots,a_n>0}\quad &J(a)=\sum_{i=1}^n w_i^{(0)}(a_i\log a_i-a_i+1),\\
\text{subject to}\quad &C(a)=\sum_{i=1}^n\{w_i^{(0)}a_i-1\}h_i=0.
\end{aligned}
\end{equation}
The target total $\sum_i h_i$ uses everyone, while only phase-2 participants have nonzero weights. For a scalar multiplier $\lambda$, the Lagrangian is
$$
L(a,\lambda)=\sum_{i=1}^n w_i^{(0)}(a_i\log a_i-a_i+1)+\lambda\sum_{i=1}^n\{w_i^{(0)}a_i-1\}h_i.
$$
For $w_i^{(0)}>0$, differentiation gives $a_i(\lambda)=\exp(-\lambda h_i)$. Values of $a_i$ for zero-weight observations do not affect the optimization. The calibration equation and its derivative are
$$
\begin{aligned}
F(\lambda)&=\sum_{i=1}^n\{w_i^{(0)}\exp(-\lambda h_i)-1\}h_i=0,\\
F'(\lambda)&=-\sum_{i=1}^n w_i^{(0)}\exp(-\lambda h_i)h_i^2.
\end{aligned}
$$
Starting at $\lambda^{(0)}=0$, the scalar Newton update is $\lambda^{(j+1)}=\lambda^{(j)}-F(\lambda^{(j)})/F'(\lambda^{(j)})$, provided the derivative is nonzero. The implementation stops when $|F(\lambda^{(j)})|<10^{-8}$ or after 2000 iterations. At a finite root $\widehat\lambda$, set $a_i^\star=\exp(-\widehat\lambda h_i)$ and $w_i^{(1)}=w_i^{(0)}a_i^\star$. These weights satisfy the calibration constraint and the stationarity conditions for the convex program, so they are optimal. Numerical termination at the iteration cap does not establish that a root was found.

For phase-2 participants, the implementation sets $\Pi_n^{(1)}(V_i)=\operatorname{clip}_{[\ell_n,1]}\{1/w_i^{(1)}\}$. It leaves the other sampling-probability entries unchanged, avoiding a ratio of two zeros for $\Delta_i=0$. The raking variant uses Algorithm~\ref{alg:ipcw_tmle_target_Pi}'s outer score check and tolerance, which remain necessary because probability bounding can alter exact calibration.

\subsection{Sampling-mechanism update}\label{sec:targeting_implementation}
In the implementation, hold the conditional regression $h_n(V)$ fixed during a sampling update. For IPCW-TMLE, it is the fitted conditional mean of the centered full-data EIC\@. For the TMLE in Section~\ref{sec:real_tmle}, it is the fitted conditional mean of $H_n(A,W)\{Y-Q_n^{(k)}(A,W)\}$ given $\Delta=1,V$. The sampling update starts at $\widetilde\Pi^{(0)}=\Pi_n^{(k)}$. At substep $j$, fit a logistic regression without an intercept, with outcome $\Delta$, offset $\operatorname{logit}\widetilde\Pi^{(j)}(V)$, and covariate $h_n(V)/\widetilde\Pi^{(j)}(V)$, using everyone. With fitted coefficient $\delta_j$, set
$$
\widetilde\Pi^{(j+1)}(V)=\operatorname{clip}_{[\ell_n,1]}\left\{\operatorname{expit}\left[\operatorname{logit}\widetilde\Pi^{(j)}(V)+\delta_j\frac{h_n(V)}{\widetilde\Pi^{(j)}(V)}\right]\right\}.
$$
The covariate is recomputed from the updated probability at the next substep.

The lower bound $\ell_n$ is supplied to the routine or defaults to $5/\{\sqrt{n_2}\log n_2\}$, where $n_2=\sum_i\Delta_i$. Stop when
$$
\left|P_n\left\{\frac{h_n(V)}{\widetilde\Pi^{(j+1)}(V)}\big(\Delta-\widetilde\Pi^{(j+1)}(V)\big)\right\}\right|<10^{-6},
$$
or after 50 substeps, and use the resulting probability estimate for the outcome update. This inner score check is distinct from the adaptive outer criterion in~\eqref{eq:adaptive_tolerance}. The default outer cap is 10 update cycles. Reaching this cap does not establish convergence.

The algorithm boxes present a $Q$-then-$\Pi$ update scheme. The implementation uses $\Pi$-then-$Q$ cycles, refits the $Q$ fluctuation through $Q_n^{(0)}$, and performs at least one cycle before checking convergence. For IPCW-TMLE, it checks the full EIC at the end of each cycle. For the TMLE in Section~\ref{sec:real_tmle}, it checks $D_Q+D_\Pi$ at the start of the next cycle. Both use the adaptive tolerance in~\eqref{eq:adaptive_tolerance}, with a configurable fixed multiplier whose default is one.

\section{Simulation setups}\label{sec:simulation_setups}
\subsection{Population targets and implementation details}\label{sec:simulation_dgps}
For the design in Section~\ref{sec:double_robust_demo}, we compute $\theta_0$ by integrating the difference between the two intervention risks over the marginal normal distribution of $0.6Z_1-0.4Z_2+0.2Z_3-0.5Z_4$. For the design in Section~\ref{sec:simulated_data}, we integrate $Q_0(1,W)-Q_0(0,W)$ over the distribution of $W$. In both designs, we compute $\Psi_{\mathrm{census}}(P_0)$ by solving the full-cohort score equations for the main-term logistic outcome regression and averaging its fitted risk difference over the covariate distribution. We use Gaussian quadrature with 30 nodes per normal covariate for these census calculations. For the variance comparison in Section~\ref{sec:simulated_data}, we compute $m_0(V)$ by integrating $D_0^F$ over $(W_3,W_4)$ conditional on $(W_1,W_2,A,Y)$ under the DGP\@. Using the same quadrature rule, we evaluate~\eqref{eq:target_pi_variance_gain} and divide by $\Var_0(U_0)$ at each sampling intercept. In the simulations, raking uses 20 imputations for its calibration variable, as described in Appendix~\ref{sec:appendix_raking}.

The simulations in Sections~\ref{sec:double_robust_demo} and~\ref{sec:simulated_data} use \texttt{R} versions 4.5.2 and 4.4.0, respectively, and \texttt{sl3} version 1.4.5. In the second experiment, the four candidates use their default hyperparameters. The learner classes are \texttt{Lrnr\_glm} for the GLM, \texttt{Lrnr\_xgboost} for boosting, \texttt{Lrnr\_ranger} for random forest, and \texttt{Lrnr\_dbarts} for BART\@.

The discrete selector, \texttt{Lrnr\_cv\_selector}, uses binomial negative log-likelihood for $Q$ and squared-error loss for the auxiliary regressions, with five folds. The separate logistic fits for $\Pi$ and $g$ use \texttt{Lrnr\_glm} through \texttt{Lrnr\_cv}. The simulation routines share the initial nuisance predictions and fold assignments across the corresponding estimator variants. These are the configurations used for the reported comparison, not a recommendation that this library suffices for every application.

\subsection{Simulations evaluating the local linearization described in Remark \ref{rem:linearization}}\label{sec:appendix_local_linearization}
Remark~\ref{rem:linearization} describes a local approximation to the uncentered full-data EIC $\bar{D}^F$ that avoids refitting its conditional regression at every iteration. Table~\ref{tab:linearization} compares this approximation with refitting and reports similar MSE and coverage in that experiment, with modest measured time savings.

\begin{table}[ht]
\centering
\resizebox{16cm}{!}{
\begin{tabular}{crcccc}
\toprule
$n$ & Approach & MSE ($\times 10^{-3}$) & Coverage (\%) & Oracle Coverage (\%) & Avg. time/run (s) \\
\midrule
\multirow{2}{*}{500}
& Re-fit $\bar{m}$ & \textbf{1.681} & 96 & 96 & \textbf{7.64} \\
& Linearization & \textbf{1.692} & 96 & 96 & \textbf{7.16} \\
\midrule
\multirow{2}{*}{1000}
& Re-fit $\bar{m}$ & \textbf{0.863} & 95 & 95 & \textbf{11.07} \\
& Linearization & \textbf{0.862} & 95 & 95 & \textbf{10.77} \\
\midrule
\multirow{2}{*}{1500}
& Re-fit $\bar{m}$ & \textbf{0.551} & 95 & 96 & \textbf{14.91} \\
& Linearization & \textbf{0.552} & 95 & 96 & \textbf{14.39} \\
\midrule
\multirow{2}{*}{2000}
& Re-fit $\bar{m}$ & \textbf{0.419} & 95 & 95 & \textbf{18.74} \\
& Linearization & \textbf{0.420} & 95 & 96 & \textbf{18.07} \\
\bottomrule
\end{tabular}
}
\caption{Simulation results comparing re-fitting $\bar{m}$ versus linearization. ``MSE'' is mean squared error; ``Coverage'' is the proportion of runs where the 95\% confidence interval based on the estimated efficient influence curve variance covers the DGP intervention contrast $\theta_0$; ``Oracle coverage'' uses the empirical SE\@. The average time per run is computed on a single-core CPU, averaged over 500 Monte-Carlo runs.}
\label{tab:linearization}
\end{table}

Table~\ref{tab:linearization_compare_positivity} examines the same approximation with extreme estimated treatment probabilities. This experiment uses the treatment, outcome, and sampling equations from Section~\ref{sec:double_robust_demo}, but sets $W_3={(0.6+Z_3Z_4/25)}^3$. We supplied the estimators with the outcome risks $P(Y=1\mid A=a,Z)$ for $a=0,1$, computed from the latent covariates $Z$. Averaging the difference between these risks over $Z$ gives the intervention risk difference $\theta_0\approx 0.244435$.

For this intervention contrast, the linearized Quasi-TMLE had an MSE of $0.1323$, compared with $0.00234$ when the conditional EIC regression was refitted. The derivative used in the linearization contains squared inverse treatment probabilities, which can make its regression on $V$ unstable.

\begin{table}[ht]
\centering
\resizebox{\linewidth}{!}{
\begin{tabular}{lrrrrrr}
\toprule
Estimator & $|\mathrm{Bias}|$ & SE & MSE & $|\mathrm{Bias}|/\mathrm{SE}$ & Coverage (\%) & Oracle coverage (\%) \\
\midrule
EEE & 0.1346 & 7.643 & 58.322 & 0.018 & 92.8 & 99.0 \\
Quasi-TMLE (linearized) & 0.1090 & 0.347 & 0.1323 & 0.314 & 94.8 & 90.4 \\
Quasi-TMLE (refitted) & 0.00133 & 0.048 & 0.00234 & 0.028 & 95.6 & 94.8 \\
\bottomrule
\end{tabular}
}
\caption{EEE and Quasi-TMLE with extreme estimated treatment probabilities, evaluated on the same 500 datasets of size $n=500$. Coverage is the percentage of 95\% intervals containing the intervention risk difference $\theta_0\approx 0.244435$. Oracle coverage uses the empirical standard error (SE) across repetitions.}
\label{tab:linearization_compare_positivity}
\end{table}

\FloatBarrier{}
\section{Proofs of lemmas}
\subsection{Proof of Lemma~\ref{lem:target_pi_variance_gain}}\label{sec:proof_target_pi_variance_gain}
\begin{proof}
Under CAR, $E_0(\Delta\mid X)=\Pi_0(V)$ and $m_0(V)=E_0(D_0^F\mid V)$. Because $D_0^F$ is the full-data EIC, $E_0(D_0^F)=0$. Conditioning on $X$ gives
$$
E_0(U_0\mid X)=D_0^F,\qquad E_0(U_0-D_0\mid X)=m_0(V)\left\{\frac{E_0(\Delta\mid X)}{\Pi_0(V)}-1\right\}=0.
$$
Taking expectations shows that $E_0(U_0)=E_0(D_0)=0$.

Since $\Delta^2=\Delta$, the moment assumption yields
$$
E_0(U_0^2)=E_0\left[\frac{{(D_0^F)}^2}{\Pi_0(V)}\right]<\infty.
$$
Conditional Jensen's inequality also gives
$$
E_0\left[\frac{{m_0(V)}^2}{\Pi_0(V)}\right]\leq E_0\left[\frac{{(D_0^F)}^2}{\Pi_0(V)}\right]<\infty.
$$
For the augmentation term, another application of $\Delta^2=\Delta$ gives
$$
\begin{aligned}
E_0\left[{(U_0-D_0)}^2\mid X\right]
&={m_0(V)}^2E_0\left[{\left\{\frac{\Delta}{\Pi_0(V)}-1\right\}}^2\mid X\right]\\
&={m_0(V)}^2\left\{\frac{\Pi_0(V)}{{\Pi_0(V)}^2}-2+1\right\}\\
&=\frac{1-\Pi_0(V)}{\Pi_0(V)}{m_0(V)}^2.
\end{aligned}
$$
Taking expectations establishes that $U_0-D_0$, and hence $D_0$, has finite variance.

To establish orthogonality, write
$$
D_0=\frac{\Delta}{\Pi_0(V)}\{D_0^F-m_0(V)\}+m_0(V).
$$
The conditional identities
$$
\begin{aligned}
E_0\left[\frac{\Delta}{\Pi_0(V)}-1\mid X\right]&=0,\\
E_0\left[\frac{\Delta}{\Pi_0(V)}\left\{\frac{\Delta}{\Pi_0(V)}-1\right\}\mid X\right]&=\frac{1-\Pi_0(V)}{\Pi_0(V)}
\end{aligned}
$$
therefore imply
$$
E_0[D_0(U_0-D_0)\mid X]=\frac{1-\Pi_0(V)}{\Pi_0(V)}m_0(V)\{D_0^F-m_0(V)\}.
$$
Taking expectations and conditioning on $V$ then yields
$$
\begin{aligned}
E_0[D_0(U_0-D_0)]
&=E_0\left[\frac{1-\Pi_0(V)}{\Pi_0(V)}m_0(V)E_0\{D_0^F-m_0(V)\mid V\}\right]\\
&=0.
\end{aligned}
$$
Finally, because $U_0$ and $D_0$ are centered,
$$
\begin{aligned}
\Var_0(U_0)-\Var_0(D_0)
&=2E_0[D_0(U_0-D_0)]+E_0[{(U_0-D_0)}^2]\\
&=E_0\left[\frac{1-\Pi_0(V)}{\Pi_0(V)}{m_0(V)}^2\right],
\end{aligned}
$$
which proves~\eqref{eq:target_pi_variance_gain}.
\end{proof}

\subsection{Proof of Lemma~\ref{lem:target_pi_approx_variance}}\label{sec:proof_target_pi_approx_variance}
\begin{proof}
By the definitions of $D_h$ and $D_0$,
$$
D_h-D_0=\left\{\frac{\Delta}{\Pi_0(V)}-1\right\}\{m_0(V)-h(V)\}.
$$
Under CAR, $E_0(\Delta\mid X)=\Pi_0(V)$ and $\Delta^2=\Delta$, so
$$
\begin{aligned}
E_0(D_h-D_0\mid X)&=0,\\
E_0[{(D_h-D_0)}^2\mid X]&=\frac{1-\Pi_0(V)}{\Pi_0(V)}{\{m_0(V)-h(V)\}}^2.
\end{aligned}
$$
The moment assumptions and the conditional Jensen bound in the proof of Lemma~\ref{lem:target_pi_variance_gain} imply that the expectation of the second expression is finite. Thus $D_h-D_0$ has mean zero and finite variance. Since $D_0$ also has mean zero and finite variance, the same holds for $D_h$.

To establish orthogonality, use
$$
D_0=\frac{\Delta}{\Pi_0(V)}\{D_0^F-m_0(V)\}+m_0(V).
$$
Multiplying by $D_h-D_0$ and taking the conditional expectation given $X$ yields
$$
E_0[D_0(D_h-D_0)\mid X]=\frac{1-\Pi_0(V)}{\Pi_0(V)}\{m_0(V)-h(V)\}\{D_0^F-m_0(V)\}.
$$
Because $m_0(V)=E_0(D_0^F\mid V)$ under CAR, conditioning on $V$ then gives
$$
\begin{aligned}
E_0[D_0(D_h-D_0)\mid V]
&=\frac{1-\Pi_0(V)}{\Pi_0(V)}\{m_0(V)-h(V)\}E_0\{D_0^F-m_0(V)\mid V\}\\
&=0.
\end{aligned}
$$
Hence $E_0[D_0(D_h-D_0)]=0$. Expanding the variance now gives
$$
\begin{aligned}
\Var_0(D_h)
&=\Var_0(D_0)+2E_0[D_0(D_h-D_0)]+E_0[{(D_h-D_0)}^2]\\
&=\Var_0(D_0)+E_0\left[\frac{1-\Pi_0(V)}{\Pi_0(V)}{\{h(V)-m_0(V)\}}^2\right].
\end{aligned}
$$
Combining this identity with~\eqref{eq:target_pi_variance_gain} yields
$$
\begin{aligned}
\Var_0(U_0)-\Var_0(D_h)
&=E_0\left[\frac{1-\Pi_0(V)}{\Pi_0(V)}{m_0(V)}^2\right]\\
&\quad-E_0\left[\frac{1-\Pi_0(V)}{\Pi_0(V)}{\{h(V)-m_0(V)\}}^2\right],
\end{aligned}
$$
which proves the stated criterion for a strict variance reduction.
\end{proof}

\subsection{Proof of Lemma \ref{lem:can_grad_of_psi_new}}
\begin{proof}
Throughout this proof, $P\in\mathcal{M}$ is an observed-data distribution, $P_X=P_X(P)$ is its recovered full-data distribution, and full-data probabilities involving $(W,A,Y)$ are evaluated under $P_X$. The conditional distribution $P_{X\mid V}=\Gamma_P$ is identified by the distribution among phase-2 participants given $V$. The proof differentiates the statistical functional $\Psi(P)=\Psi^F(P_X(P))$ and does not require a causal interpretation.

Let's first focus on the treatment-specific mean part of the target parameter. Suppose the joint density $p(w_1,w_2)$ is treated as given. Then, our target parameter (treatment-specific mean part) mapping can be represented as 
$$
P\rightarrow\int_{w_1,w_2}\bar{Q}(w_1,w_2,1)\,dP(w_1,w_2),
$$
where $\bar{Q}(w_1,w_2,a)=E(Y\mid W_1=w_1,W_2=w_2,A=a)$. We now find the efficient influence curve of this parameter mapping. For simplicity, let's first focus on the case when $Y$ is binary. Then, the parameter is
$$
P\rightarrow\int_{w}\frac{P(Y=1,w,A=1)}{P(w,A=1)}\,dP(w).
$$
Note that since $P(Y=1,w,A=1)$ is a mean parameter, it has efficient influence curve $D^F_{Y,w,1}(X)=I(Y=1,W=w,A=1)-P(Y=1,W=w,A=1)$. Similarly, $P(W=w,A=1)$ has efficient influence curve $D^F_{W,1}(X)=I(W=w,A=1)-P(W=w,A=1)$. Recall also that the A-IPCW mapping applied to $D^F_{Y,W,1}$ and $D^F_{W,1}$ are
\begin{align*}
D^\star_{Y,w,1}&=\frac{\Delta}{\Pi(V)}D^F_{Y,w,1}(X)-\frac{E(D^F_{Y,w,1}(X)\mid\Delta=1,V)}{\Pi(V)}(\Delta-\Pi(V))\quad\text{and}\\
D^\star_{w,1}&=\frac{\Delta}{\Pi(V)}D^F_{w,1}(X)-\frac{E(D^F_{w,1}(X)\mid\Delta=1,V)}{\Pi(V)}(\Delta-\Pi(V)).
\end{align*}
By the delta-method, the EIC of $P\rightarrow\int_{w}P(Y=1,w,A=1)/P(w,A=1)\,dP(w)$ treating $p(w_1,w_2)$ as given is then
\begin{align*}
D_Q=\int_w\frac{1}{P(W=w,A=1)}D^\star_{Y,w,1}-\frac{P(Y=1,W=w,A=1)}{[P(W=w,A=1)]^2}D^\star_{w,1}\,dw
\end{align*}
Applying the A-IPCW mapping we would end up with four terms, we denote them as terms $\mathbf{A}$, $\mathbf{B}$, $\mathbf{C}$, and $\mathbf{D}$. Specifically,
\begin{align*}
\mathbf{A}&=\int\frac{1}{P(A=1,W=w)}\cdot\frac{\Delta}{\Pi(V)}[I(Y=1,A=1,W)-P(Y=1,A=1,W=w)]\,dP(w)\\
\mathbf{B}&=-\int\frac{1}{P(A=1,W=w)}\cdot\frac{\Delta-\Pi(V)}{\Pi(V)}E[I(Y=1,A=1,W=w)-P(Y=1,A=1,W=w)\mid\Delta=1,V]\,dP(w)\\
\mathbf{C}&=-\int\frac{P(Y=1,A=1,W=w)}{[P(A=1,W=w)]^2}\cdot\frac{\Delta}{\Pi(V)}[I(A=1,W=w)-P(A=1,W=w)]\,dP(w)\\
\mathbf{D}&=\int\frac{P(Y=1,A=1,W=w)}{[P(A=1,W=w)]^2}\cdot\frac{\Delta-\Pi(V)}{\Pi(V)}E[I(A=1,W=w)-P(A=1,W=w)\mid\Delta=1,V]\,dP(w)
\end{align*}
For $\mathbf{A}$, we have
\begin{align*}
\mathbf{A}&=\int\frac{\Delta}{\Pi(V)}\cdot\frac{I(Y=1,A=1,W=w)}{P(A=1,W=w)}-\frac{\Delta}{\Pi(V)}\cdot\frac{P(Y=1,A=1,W=w)}{P(A=1,W=w)}\,dP(w)\\
&=\frac{\Delta}{\Pi(V)}\left[\frac{I(Y=1,A=1)}{P(A=1\mid W)}-\Psi(P)\right].
\end{align*}
For $\mathbf{B}$, we have
\begin{align*}
\mathbf{B}&=-\frac{\Delta-\Pi(V)}{\Pi(V)}\left[\int\frac{E[I(Y=1,A=1,W=w)\mid\Delta=1,V]}{P(A=1,W=w)}-\frac{P(Y=1,A=1,W=w)}{P(A=1,W=w)}\,dP(w)\right]\\
&=-\frac{\Delta-\Pi(V)}{\Pi(V)}\left[E\left(\frac{I(A=1)\cdot Y}{P(A=1\mid W)}\middle\vert\Delta=1,V\right)-\Psi(P)\right]
\end{align*}
For $\mathbf{C}$, we have
\begin{align*}
\mathbf{C}&=-\frac{\Delta}{\Pi(V)}\left[\int\frac{P(Y=1,A=1,W=w)}{[P(A=1,W=w)]^2}I(A=1,W=w)-\frac{P(Y=1,A=1,W=w)}{P(A=1,W=w)}\,dP(w)\right]\\
&=-\frac{\Delta}{\Pi(V)}\left[\frac{I(A=1)}{P(A=1\mid W)}Q(1,W)-\Psi(P)\right]
\end{align*}
For $\mathbf{D}$, we have
\begin{align*}
\mathbf{D}&=\frac{\Delta-\Pi(V)}{\Pi(V)}\left[E\left(\frac{I(A=1)}{P(A=1\mid W)}Q(1,W)\mid\Delta=1,V\right)-\Psi(P)\right]
\end{align*}
To summarize, the EIC for the parameter $\Psi$
$$
D^\star=\frac{\Delta}{\Pi(V)}\left[\frac{I(A=1)}{g(W)}(Y-Q(1,W))\right]-\frac{\Delta-\Pi(V)}{\Pi(V)}E\left(\frac{I(A=1)}{g(W)}(Y-Q(1,W))\middle\vert\Delta=1,V\right).
$$
We note that this is just the A-IPCW mapping applied to the $P_Y$-score component of the EIC.

Next, we find the influence curve treating $Q$ and $P_V$ as known. Consider the equivalent mapping
$$
P\rightarrow\frac{E(\Delta Q(1,W)\mid\Delta=1,V)}{P(\Delta=1\mid V)}.
$$
Note that the influence curve of the numerator is $\Delta Q(1,W)-E(\Delta Q(1,W)\mid V)$, and the influence curve of the denominator is $\Delta-P(\Delta=1\mid V)$. By the delta-method, the influence curve of the ratio of the numerator and denominator is
\begin{align*}
\text{IC}_{P_{X\mid V}}&=\frac{1}{P(\Delta=1\mid V)}[\Delta Q(1,W)-E(\Delta Q(1,W)\mid V)]-\frac{E(\Delta Q(1,W)\mid V)}{[P(\Delta=1\mid V)]^2}[\Delta-P(\Delta=1\mid V)]\\
&=\frac{\Delta}{\Pi(V)}[Q(1,W)-E(Q(1,W)\mid\Delta=1,V)]
\end{align*}

Next, we find the influence curve of the $P_V$-score component treating $Q$ and $P_{X\mid V}$ as fixed, which is given by $E(Q(1,W)\mid\Delta=1,V)-E_{P_V}E(Q(1,W)\mid\Delta=1,V)$.

The EIC derivation of the control ($A=0$) portion of the target parameter mirrors. To summarize, the EIC of the target parameter mapping
$$
P\rightarrow E_{P_V}E_{P_{X\mid V}}[Q_{P_X}(1,W)-Q_{P_X}(0,W)\mid\Delta=1,V]
$$
is
$$
D^\star=D_{Q}+D_{P_{X\mid V}}+D_{P_V},
$$
where
\begin{align*}
D_{Q}&=\frac{\Delta}{\Pi_P(V)}\left[H_{g,P_X}(A,W)(Y-Q_{P_X}(A,W))\right]-\frac{\Delta-\Pi_P(V)}{\Pi_P(V)}E\left[H_{g,P_X}(A,W)(Y-Q_{P_X}(A,W))\right],\\
D_{P_{X\mid V}}&=\frac{\Delta}{\Pi_P(V)}\left[Q_{P_X}(1,W)-Q_{P_X}(0,W)-E_{P_{X\mid V}}[Q_{P_X}(1,W)-Q_{P_X}(0,W)\mid\Delta=1,V]\right],\\
D_{P_V}&=E_{P_{X\mid V}}[Q_{P_X}(1,W)-Q_{P_X}(0,W)\mid\Delta=1,V]-\Psi(P),
\end{align*}
and where 
$$
H_{g,P_X}(A,W)=\frac{A}{g_{P_X}(W)}-\frac{1-A}{1-g_{P_X}(W)}.
$$
\end{proof}

\subsection{Proof of Lemma \ref{lem:R2_new}}
\begin{proof}
Let $D^F_{Q_P}=H_{g_P}(A,W)(Y-Q_P(A,W))$, where $H_{g_P}=A/g_P(W)-(1-A)/\{1-g_P(W)\}$. We first consider the first two components of the canonical gradient:
$$
D_{Q,P}-D_{\Pi,P}=\frac{\Delta}{\Pi_P(V)}D^F_{Q,P}-\frac{\Delta-\Pi_P(V)}{\Pi_P(V)}E_\Gamma D^F_{Q,P}.
$$
Note that
\begin{align*}
P_0\left\{D_{Q,P}-D_{\Pi,P}\right\}&=P_0\left\{\frac{\Delta}{\Pi_P(V)}D^F_{Q,P}-\frac{\Delta-\Pi_P(V)}{\Pi_P(V)}E_{P}(D^F_{Q,P}\mid\Delta=1,V)\right\}\\
&=P_0\left\{\frac{\Delta}{\Pi_P(V)}\left[D^F_{Q,P}-E_{P}(D^F_{Q,P}\mid\Delta=1,V)\right]+E_{P}(D^F_{Q,P}\mid\Delta=1,V)\right\}\\
&=P_0\left\{\frac{\Pi_0}{\Pi_P}(V)\left[E_{P_0}(D^F_{Q,P}\mid\Delta=1,V)-E_{P}(D^F_{Q,P}\mid\Delta=1,V)\right]+E_{P}(D^F_{Q,P}\mid\Delta=1,V)\right\}\\
&=P_0\left\{\frac{\Pi_P-\Pi_0}{\Pi_P}(V)\left[E_P(D^F_{Q,P}\mid\Delta=1,V)-E_{0}(D^F_{Q,P}\mid\Delta=1,V)\right]+\underbrace{E_{0}(D^F_{Q,P}\mid\Delta=1,V)}_{\mathbf{A}}\right\}.
\end{align*}
For term $\mathbf{A}$, we have that
\begin{align*}
P_0\mathbf{A}&=E_{P_{X,0}}\left\{\{g_P(W)-g_0(W)\}\left[\frac{(Q_P-Q_0)(1,W)}{g_P(W)}+\frac{(Q_P-Q_0)(0,W)}{1-g_P(W)}\right]\right\}\\
&-P_0\left\{(Q_P-Q_0)(1,W)\mid\Delta=1,V\right\}+P_0\left\{(Q_P-Q_0)(0,W)\mid\Delta=1,V\right\}.
\end{align*}
Let $\mu^{(a)}_{P_1,P_2}=E_{P_1}[Q_{P_2}(a,W)\mid\Delta=1,V]$. For the treatment arm part of $D_{\Gamma,P}$, we have that
\begin{align*}
P_0\left\{\frac{\Delta}{\Pi_P(V)}\left[Q_P(1,W)-\mu_{P,P}^{(1)}\right]\right\}&=P_0\left\{\frac{\Pi_0}{\Pi_P}(V)\left[\mu_{P_0,P}^{(1)}-\mu_{P,P}^{(1)}\right]\right\}\\
&=P_0\left\{\frac{\Pi_P-\Pi_0}{\Pi_P}(V)\left[\mu_{P,P}^{(1)}-\mu_{P_0,P}^{(1)}\right]+\mu_{P_0,P}^{(1)}-\mu_{P,P}^{(1)}\right\}.
\end{align*}
Similarly, for the control arm part of $D_{\Gamma,P}$, we have
$$
-P_0\left\{\frac{\Pi_P-\Pi_0}{\Pi_P}(V)\left[\mu_{P,P}^{(0)}-\mu_{P_0,P}^{(0)}\right]+\mu_{P_0,P}^{(0)}-\mu_{P,P}^{(0)}\right\}.
$$
Finally, for the last piece of the canonical gradient:
$$
D_{P_V,P}=E_P[Q_P(1,W)-Q_P(0,W)\mid\Delta=1,V]-\Psi(P),
$$
we have that
$$
P_0D_{P_V,P}=P_0\left\{\mu_{P,P}^{(1)}-\mu_{P,P}^{(0)}\right\}-\Psi(P).
$$
Therefore, to summarize, the exact remainder is given by
\begin{align*}
R(P,P_0)&=P_0\left\{\frac{\Pi_P-\Pi_0}{\Pi_P}(V)\left[E_P(D^F_{Q,P}\mid\Delta=1,V)-E_{0}(D_{Q,P}^F\mid\Delta=1,V)\right]\right\}\\
&+E_{P_{X,0}}\left\{\{g_P(W)-g_0(W)\}\left[\frac{(Q_P-Q_0)(1,W)}{g_P(W)}+\frac{(Q_P-Q_0)(0,W)}{1-g_P(W)}\right]\right\}\\
&+P_0\left\{\frac{\Pi_P-\Pi_0}{\Pi_P}(V)\left(\mu_{P,P}^{(1)}-\mu_{P_0,P}^{(1)}\right)-\frac{\Pi_P-\Pi_0}{\Pi_P}(V)\left(\mu_{P,P}^{(0)}-\mu_{P_0,P}^{(0)}\right)\right\}.
\end{align*}
\end{proof}

\section{A particular TMLE for CAR censored data estimation problems.}
Let $X\sim P_{X,0}\in {\cal M}^F$ and our target parameter is $\Psi^F:{\cal M}^F\rightarrow\RR$.
Suppose $\Psi^F(P_X)=\Psi^F(Q_X)$ for some part $Q(P_X)$ of the full-data distribution, and let $L^F(Q)$ be a full-data loss function for learning $Q(P_X)$. We will also denote $Q_X$ with $Q$ and $Q_{X,0}$ with $Q_0$. Let $D^F_{Q,G}$ be the canonical gradient of this full-data parameter, possibly depending on another nuisance parameter $G(P_X)$ beyond $Q$, and let $R^F((Q,G),(Q_0,G_0))=\Psi^F(Q)-\Psi^F(Q_0)+E_{P_{X,0}}D^F_{Q,G}$ be the corresponding exact remainder. We only observe $n$ i.i.d.\ copies of $O=\Phi(C,X)\sim P_{P_{X,0},\Pi_0}$ and  the conditional distribution $\Pi$ of $C$, given $X$, satisfying CAR\@. Let ${\cal M}=\{P_{P_X,\Pi}:P_X\in {\cal M}^F,\Pi\in\mathcal{G}_C\}$ be the observed data model implied by the full-data model and censoring mechanism model $\mathcal{G}_C$. Let $\Psi:{\cal M}\rightarrow\RR$ be defined by 
$\Psi(P_{P_X,\Pi})=\Psi^F(Q(P_X))$, assuming identifiability of $\Psi^F(P_X)$ from the observed-data distribution. This generalizes the correspondence in Section~\ref{sec:notation} to coarsening patterns $C$ beyond two-phase sampling. Here $\mathcal{G}_C$ denotes the model for the coarsening mechanism, distinct from the conditional distribution $\Gamma_P$ used in the main text.
Let $D_{P_X,\Pi,D^F_{Q,G}}$ be the canonical gradient of $\Psi$ at $P_{P_X,\Pi}$ and let $R((P_X,\Pi),(P_{X,0},\Pi_0))=\Psi^F(Q)-\Psi^F(Q_0)+P_{P_{X,0},\Pi_0}D_{P_X,\Pi,D^F_{Q,G}}$ be the exact remainder.
There is a rich literature on CAR censored data models including an A-IPCW representation of the canonical gradient introduced by \cite{robins_estimation_1994,robins_analysis_1995} (\cite{vdl_unified_2003}).

{\bf Representation of the observed data model canonical gradient:}
Let $T(P_X)$ be the tangent space of the full-data model at $P_X$. 
Let $A_{P_X}:T(P_X)\rightarrow L^2_0(P_{P_X,\Pi})$ be the score operator defined by
$A_{P_X}(h)=E_{P_X}(h(X)\mid O)$. Let $A^\star_{\Pi,\text{np}}:L^2_0(P_{P_X,\Pi})\rightarrow L^2_0(P_X)$ be the nonparametric adjoint defined by 
$A^\star_{\Pi,\text{np}}(V)=E(V(O)\mid X)$. Let $\Pi_{T(P_X)}$ be the projection operator onto $T(P_X)$ in $L^2_0(P_X)$. 
The adjoint of $A_{P_X}$ is then given by $A_{P_X,\Pi}^\star\equiv \Pi_{T(P_X)}A_{\Pi,\text{np}}A_{P_X}$. If the full-data model is nonparametric then $T(P_X)=L^2_0(P_X)$ and $\Pi_{T(P_X)}$ is the identity operator. 
Then, we can define the information operator $I_{P_X,\Pi}:T(P_X)\rightarrow T(P_X)$ defined by 
$I_{P_X,\Pi}=A_{P_X,\Pi}^\star A_{P_X}: T(P_X)\rightarrow T(P_X)$. Let $I_{P_X,\Pi,np}=A_{\Pi}^\star A_{P_X}$ be the nonparametric information operator. The canonical gradient can be represented as
\[
D_{P_X,\Pi,D^F_{Q,G}}=A_{P_X}I_{P_X,\Pi}^{-1}(D^F_{Q,G}).\]
One might call this the least-squares projection representation of the canonical gradient.
The A-IPCW representation applied to a full-data function is of the form
\[
D^{\text{A-IPCW}}_{P_X,\Pi,\tilde{D}^F}=U_{\Pi}(\tilde{D}^F)-\Pi(U_{\Pi}(\tilde{D}^F)\mid T_{\text{CAR}}(\Pi)),\]
where $T_{\text{CAR}}(\Pi)=\{S\in L^2_0(P_{P_X,\Pi}): E(S\mid X)=0\}$ is the tangent space of $\Pi$ under CAR only,
and $U_{\Pi}(\tilde{D}^F)$ is an inverse probability of censoring weighted full-data function satisfying
$E_{\Pi}(U_{\Pi}(\tilde{D}^F)\mid X)=\tilde{D}^F(X)$. 
If the full-data model is nonparametric, then $D^{\text{A-IPCW}}_{P_X,\Pi,D^F_{Q,G}}=D_{P_X,\Pi,D^F_{Q,G}}$, i.e., the A-IPCW mapping applied to the full-data canonical gradient yields the canonical gradient. 
However, if the full-data tangent space is not saturated, then one needs to select $\tilde{D}^F=I_{P_X,\Pi,\text{np}}I_{P_X,\Pi}^{-1}(D^F_{Q,G})$ in order to get the canonical gradient, which  is as complicated as the canonical gradient itself. This makes the A-IPCW representation not helpful when the full-data model is not saturated.

{\bf Full-data TMLE:} Let $\{Q_{X,e}:e\}$ be a least favorable path in the full-data model so that $d/de_0 L(Q_{X,e_0})=D^F_{Q_X}$ at $e=e_0=0$.
Given an initial estimator $Q_{n}$ of $Q_0$, in the full-data model, we could then use the TMLE defined by 
$\Psi^F(Q_n^\star)$ with $e_n^F=\arg\min_e P_{X,n} L^F(Q_{n,e})$ and $Q_n^\star=Q_{n,e_n^F}$, $P_{X,n}$ the empirical distribution of full-data, where one might possibly have to iterate a few times or use a universal LFM.  One would then have solved
$0=P_{X,n} D^F_{Q_n^\star,G_n}$.

{\bf An observed-data TMLE based on the full-data risk:}
This construction uses the same full-data LFM and estimates $e_n$ by optimizing an efficient estimator of the full-data risk $P_{X,0}L^F(Q_{n,e})$. It follows the full-data TMLE update with an estimated risk replacing the unavailable full-data empirical risk, so that $e_n$ approximates $e_n^F$.
For a given $Q_{n,e}$ we define the full-data risk/parameter $R^F_{Q_{n,e}}(P_X)\equiv P_X L^F(Q_{n,e})$.
Let $D^F_{R^F_{Q_{n,e}}(\cdot),P_X}$ be the full-data model canonical gradient of this full-data parameter.
Typically, we have that $D^F_{R^F_{Q_{n,e}}(\cdot),P_X}=L^F(Q_{n,e})-R^F_{Q_{n,e}}(P_X)$, just the full-data loss function minus its expectation. Let's assume this scenario, which is always true when the full-data model is nonparametric, but typically also holds for other full-data models.
Then $D_{P_X,\Pi,D^F_{R^F_{Q_{n,e}}(\cdot),P_X} }$ is the canonical gradient of $R^F_{Q_{n,e}}(\cdot)$ at $P_{P_X,\Pi}$.
Recall   \[
D_{P_X,\Pi,D^F_{R^F_{Q_{n,e}}(\cdot),P_X}  } =A_{P_X}I_{P_X,\Pi}^{-1}(L^F(Q_{n,e}))- P_X L^F(Q_{n,e}).\]
For each $e$, consider a TMLE $P_{X,n}^{e,\star}$ targeting $P_X L^F(Q_{n,e})$, so that 
\[ 
0=P_n A_{P_{X,n}^{e,\star}}I_{P_{X,n}^{e,\star},\Pi_n}^{-1}\left(L^F(Q_{n,e})\right)-R^F_{Q_{n,e}}(P_{X,n}^{e,*}),\]
which implies that the plug-in TMLE $R^F_{Q_{n,e}}(P_{X,n}^{e,\star})$ of the full-data risk at $Q_{n,e}$ is given by
\[
R^F_{Q_{n,e}}(P_{X,n}^{e,\star})=P_n A_{P_{X,n}^{e,\star}}I_{P_{X,n}^{e,\star},\Pi_n}^{-1}(L^F(Q_{n,e}) ).\]
Let $e_n^0=0$, $k=0$, and iterate
\[
e_n^{k+1}=\arg\min_e R^F_{Q_{n,e}}(P_{X,n}^{e_n^k}).\]
At convergence we have $e_n^{k+1}\approx e_n^k$. 
Then, $Q_n^\star=Q_{n,e_n}$ solves its derivative equation given by
\[
0=P_n A_{P_{X,n}^{e_n,\star}}I_{P_{X,n}^{e_n,\star},\Pi_n}^{-1}\left( d/de_n L^F(Q_{n,e_n})\right ).\]
If $Q_{n,e}$ is a universal full-data model LFM, then we have
\[
d/de_nL^F(Q_{n,e_n})=D^F_{Q_{n,e_n},G_n}.\]
For a general local LFM, we can iterate the above by replacing the off-set $Q_n$ by $Q_{n,e_n}$ and iterate till convergence at which $e_n\approx 0$.
Thus, with $P_{X,n}^*=P_{X,n}^{e_n,\star}$, we have
\[
0=P_n A_{P_{X,n}^\star}I_{P_{X,n}^\star,\Pi_n}^{-1}(D^F_{Q_{n,e_n},G_n}).\]
Thus, $Q_n^\star=Q_{n,e_n}$ solves the efficient influence curve score equation so that $\Psi^F(Q_{n,e_n})$ can be analyzed as a TMLE making it an efficient estimator under standard regularity conditions. 

{\bf Non-TMLE analogue:}
Suppose now that we just use the initial estimator $P_{X,n}$ and $\Pi_n$ to construct
a one-step estimator or, equivalently, A-IPCW estimator of $R^F_{Q_{n,e}}(P_X)$ for all $e$, given by:
\[
R^F_{Q_{n,e}}(P_{X,n},\Pi_n,P_n)\equiv P_n A_{P_{X,n}}I_{P_{X,n},\Pi_n}^{-1}(L^F(Q_{n,e})).\]
Let
\[
e_n=\arg\min_e R^F_{Q_{n,e}}(P_{X,n},\Pi_n,P_n).\]
Then, $Q_n^\star=Q_{n,e_n}$ solves its derivative equation given by
\[
0=P_n A_{P_{X,n}}I_{P_{X,n},\Pi_n}^{-1}(d/de_n L^F(Q_{n,e_n})).\]
If $Q_{n,e}$ is a universal full-data model  LFM, then we have
\[
d/de_nL^F(Q_{n,e_n})=D^F_{Q_{n,e_n},G_n}.\]
So then
\[
0=P_n A_{P_{X,n}}I_{P_{X,n},\Pi_n}^{-1}(D^F_{Q_{n,e_n},G_n}).\]
This can now be analyzed as follows.
We have
\[
\begin{array}{l}
-P_0 A_{P_{X,n}}I_{P_{X,n},\Pi_n}^{-1}(D^F_{Q_{n,e_n},G_n})=
(P_n-P_0) A_{P_{X,n}}I_{P_{X,n},\Pi_n}^{-1}(D^F_{Q_{n,e_n},G_n}).
\end{array}
\]
The right-hand side is an empirical process term so that under the assumption that $P_{X,n},\Pi_n$ are consistent for $P_{X,0}$, $\Pi_0$, and $Q_n$ is consistent for $Q_0$, and a Donsker class condition: 
\[
R_{n,1}\equiv (P_n-P_0) A_{P_{X,n}}I_{P_{X,n},\Pi_n}^{-1}(D^F_{Q_{n,e_n},G_n})-P_n A_{P_{X,0}}I_{P_{X,0},\Pi_0}^{-1}(D^F_{Q_0,G_0})=o_P(n^{-1/2}),
\]
providing the desired efficient sample mean approximation. 
We also note that if $\Pi_n=\Pi_0$  but $P_{X,n}$ is inconsistent converging to misspecified $P_{X,1}$, while $D^F_{Q_{n,e_n},G_n}$ is still consistent for an unbiased full-data function $D^F_0$ such as $D^F_{Q_0,G_0}$ (but underlying problem might also have double-robust structure so that even in $D^F$ there might be inconsistent estimation of full-data distribution), then we can still obtain asymptotic linearity with influence curve
$A_{P_{X,1}}I_{P_{X,1},\Pi_0}^{-1}(D^F_0)$.
This can be further extended to allow $\Pi_n$ is a consistent estimator of $\Pi_0$ such as HAL and $P_{X,n}$ is inconsistent.
Let's consider the case that $P_{X,n},\Pi_n$ are both consistent. 

In addition, we have
\begin{align*}
-P_0 A_{P_{X,n}}I_{P_{X,n},\Pi_n}^{-1}(D^F_{Q_{n,e_n},G_n})&=
-P_0 \left\{ A_{P_{X,n}}I_{P_{X,n},\Pi_n}^{-1}(D^F_{Q_{n,e_n},G_n})-
A_{P_{X,n}}I_{P_{X,n},\Pi_0}^{-1}(D^F_{Q_{n,e_n},G_n})\right\} \\
&-P_0 A_{P_{X,n}}I_{P_{X,n},\Pi_0}^{-1}(D^F_{Q_{n,e_n},G_n}).
\end{align*}
Using that $P_0 A_{P_{X,0}}I_{P_{X,0},\Pi_1}^{-1}(D^F)=P_{X,0}D^F$ for any $\Pi_1$, the first term is a second-order remainder by noting that it equals
\begin{align*}
R_{n,2}&\equiv -P_0 \left\{ A_{P_{X,n}}I_{P_{X,n},\Pi_n}^{-1}(D^F_{Q_{n,e_n},G_n})-
A_{P_{X,n}}I_{P_{X,n},\Pi_0}^{-1}(D^F_{Q_{n,e_n},G_n})\right\} \\
&+P_0\left\{ A_{P_{X,0}}I_{P_{X,0},\Pi_n}^{-1}(D^F_{Q_0,G_0})-
A_{P_{X,0}}I_{P_{X,0},\Pi_0}^{-1}(D^F_{Q_0,G_0})\right\}\\
&+P_{X,0}D^F_{Q_{n,e_n},G_n}
\end{align*}
This is now a second-order difference in $(P_{X,n}-P_{X,0},D^F_{Q_n,G_n}-D^F_{Q_0,G_0})$ and $\Pi_n-\Pi_0$, so that  $R_{n,2}=o_P(n^{-1/2})$ under rate conditions.

Due to $E(A_{P_{X,1}}I_{P_{X,1},\Pi_0}^{-1}(D^F))=E_{P_{X,0}}D^F$, for any $P_{X,1}$, we have
\[
-P_0 A_{P_{X,n}}I_{P_{X,n},\Pi_0}^{-1}(D^F_{Q_{n,e_n},G_n})=-P_{X,0}(D^F_{Q_{n,e_n},G_n}).\]
By definition of the exact second-order remainder in the full-data model the right-hand side equals
\[
\Psi^F(Q_{n,e_n})-\Psi^F(Q_0)-R^F((Q_{n,e_n},G_n),(Q_0,G_0)).\]
Combining all these yields the following expansion:
\[
\begin{array}{l}
\Psi^F(Q_{n,e_n})-\Psi^F(Q_0)=P_n A_{P_{X,0}}I_{P_{X,0},\Pi_0}^{-1}(D^F_{Q_0,G_0})
+R_{n,1}+R_{n,2}+R^F(Q_{n,e_n},G_n,Q_0,G_0).
\end{array}
\]
So, if $R^F((Q_{n,e_n},G_n),(Q_0,G_0))=o_P(n^{-1/2})$ also, then it follows that $\Psi^F(Q_{n,e_n})$ is asymptotically efficient. 

The latter non-TMLE efficient estimator of the full-data risk is easier and avoids the need for iteration, while we still obtain a plug-in estimator $\Psi^F(Q_{n,e_n})$. Thus, one might consider estimation of the full-data risk with A-IPCW type estimators without being concerned with using a TMLE.

\textbf{Comparison with IPCW-full-data TMLE:} The IPCW-TMLE with targeting of $\Pi$ also uses an estimator of the full-data risk of $Q_{n,e}$ and minimizes that. By targeting the censoring mechanism in this estimator, we obtain an efficient estimator of the full-data risk if the full-data model is nonparametric. In that case our estimator above is similar to the IPCW-TMLE with targeting of $\Pi$. However, it clarifies that we can use any efficient estimator of the full-data risk and that for the sake of targeting this is fine, even though for constructing an initial estimator the IPCW loss is very convenient. Moreover, if the full-data model is not saturated the targeted IPCW estimator of the full-data risk is still inefficient. We might have called the above TMLE a calibrated full-data TMLE since it uses the same form estimator $\Psi^F(Q_{n,e_n})$ but it replaces the $e_n^F$ in the full-data TMLE by the minimizer of an efficient estimator of the full-data risk: so the difference is $\Psi^F(Q_{n,e_n})-\Psi^F(Q_{n,e_n^F})$.
Therefore, we could the extra work going into  selecting $e_n$ as a calibration of the full-data TMLE taking into account the censoring of the full-data.

\end{document}